\documentclass[journal]{IEEEtran}

\usepackage{subcaption}
\usepackage{tabularx}
\usepackage{booktabs}
\usepackage{xcolor}

\ifCLASSINFOpdf
  \usepackage[pdftex]{graphicx}
\else
  \usepackage[dvips]{graphicx}
\fi
\usepackage{amsmath}
\usepackage{mathtools}
\usepackage{amsthm}

\newtheorem{theorem2}{Theorem2}[section]

\newtheorem{assumption}[theorem2]{Assumption}
\newtheorem{prop}{Proposition}

\usepackage{amsfonts}
\usepackage{bbm}
\usepackage[ruled,vlined]{algorithm2e}

\usepackage{url}
\usepackage{color}
\definecolor{myGreen}{RGB}{0, 0, 0}
\definecolor{myBlue}{RGB}{0, 0, 0}

\begin{document}
%
\title{Proof of Travel for Trust-Based Data Validation in V2I Communication}
%
%
%

\author{Dajiang~Suo, Baichuan~Mo,
        Jinhua Zhao,
        and~Sanjay~E.~Sarma
\thanks{D. Suo and S. Sarma are with the Department of Mechanical Engineering, Massachusetts Institute of Technology, Cambridge, MA, 02139. USA e-mail: {djsuo, sesarma}@mit.edu.}
\thanks{B. Mo is with the Department of Civil and Environmental Engineering, Massachusetts Institute of Technology, Cambridge,
MA, 02139. USA e-mail: {baichuan}@mit.edu.}
\thanks{J. Zhao is with the Department of Urban Studies and Planning, Massachusetts Institute of Technology, Cambridge,
MA, 02139. USA e-mail: {jinhua}@mit.edu.}
\thanks{Corresponding author: Dajiang Suo (djsuo@mit.edu)}}

%
%

\markboth{This paper has been accepted for publication in IEEE Internet of Things Journal}%
{Shell \MakeLowercase{\textit{et al.}}: Bare Demo of IEEEtran.cls for IEEE Journals}
%



\maketitle

\begin{abstract}
Previous work on misbehavior detection and trust management for Vehicle-to-Everything (V2X) communication security is effective in identifying falsified and malicious V2X data. Each vehicle in a given region can be a witness to report on the misbehavior of other nearby vehicles, which will then be added to a "blacklist." However, there may not exist enough witness vehicles that are willing to opt-in in the early stage of connected-vehicle deployment. In this paper, we propose a "whitelisting" approach to V2X security, titled Proof-of-Travel (POT), which leverages the support of roadside infrastructure. Our goal is to transform the power of cryptography techniques embedded within Vehicle-to-Infrastructure (V2I) protocols into game-theoretic mechanisms to incentivize connected-vehicle data sharing and validate data trustworthiness simultaneously.

The key idea is to determine the reputation of and the contribution made by a vehicle based on its distance traveled and the information it shared through V2I channels. In particular, the total vehicle miles traveled for a vehicle must be testified by digital signatures signed by each infrastructure component along the path of its movement. While building a chain of proofs of spatial movement creates burdens for malicious vehicles, acquiring proofs does not result in extra costs for normal vehicles, which naturally want to move from the origin to the destination. The POT protocol is used to enhance the security of previous voting-based data validation algorithms for V2I crowdsensing applications. For the POT-enhanced voting, we prove that all vehicles choosing to cheat are not a pure Nash equilibrium using game-theoretic analysis. Simulation results suggest that the POT-enhanced voting is more robust to malicious data.
\end{abstract}


%
\IEEEpeerreviewmaketitle

\section{Introduction}
%
%
%
%

\IEEEPARstart{V}{ehicle}-to-Infrastructure (V2I) communication technologies show promises in improving traffic in urban areas during rush hours~\cite{johnson2016connected} and assisting the transportation management center (TMC) in responding to emergency situations~\cite{CVNYC,CAVEmergencyNeeds,kitchener2018connected}. While the V2I data about vehicle movement and status can be used for signal control to ameliorate traffic congestions, the data containing sensing information about road emergencies (e.g., incidents or work zones) can help the TMC and emergency responders allocate rescuing resources, plan routes~\cite{kitchener2018connected}, and determine if it is necessary to create geofences around incident sites~\cite{chowdhury2018lessons,EmergencyNotification}. Previous studies show that information about road incidents sent from vehicles to public safety answering points (PSAP) contributes to reducing the response time of emergency vehicles for achieving a low mortality rate~\cite{obenauf2019impact}.

For a traffic event to be broadcasted through V2I channels, the information regarding the time, location, severity, and surrounding environment of the event is crucial for the TMC to make correct decisions~\cite{rostamzadeh2015context}. However, it is challenging for a roadside unit (RSU) to verify a nearby V2I event in real-time before disseminating it to the TMC due to the existence of malicious vehicle nodes. For example, an adversary who holds valid vehicle credentials can fabricate vehicle identities and have compromised vehicles send fake V2I messages to infrastructure~\cite{feng2018vulnerability,islam2018cybersecurity,suo2020location}. 

\begin{figure*}[tb!]
\centering
\includegraphics[width=0.9\textwidth]{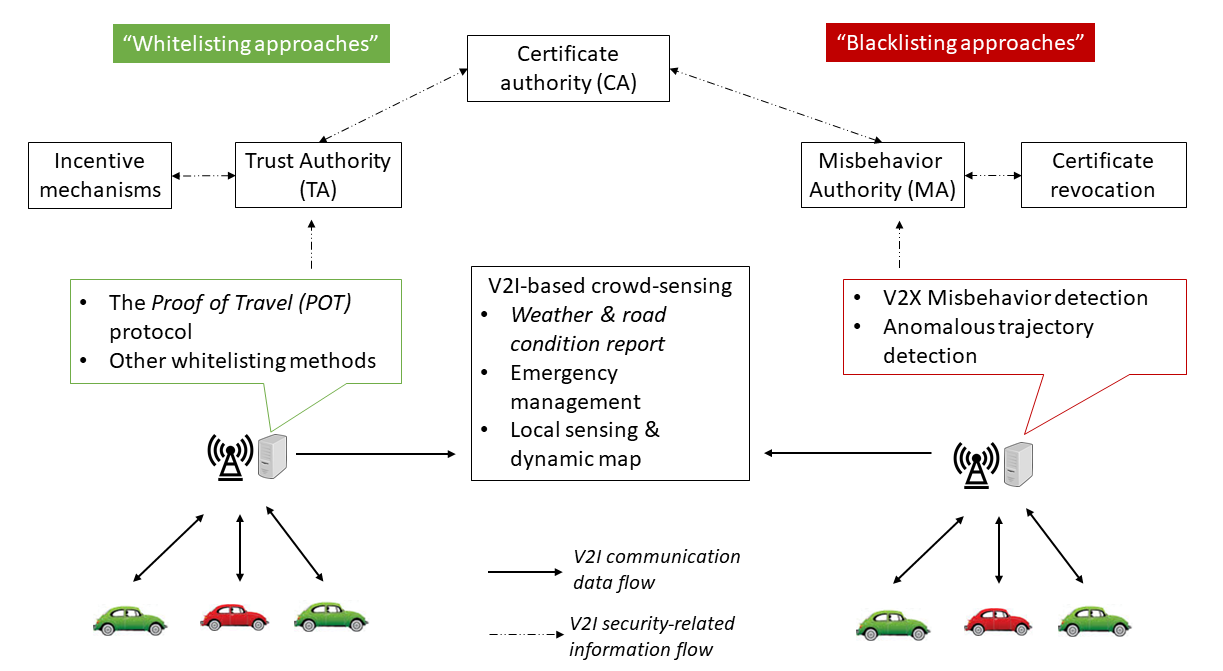}
\caption{The role of the Proof of Travel protocol in the V2I security landscape.}
\label{fig1}
\end{figure*}

Previous work on V2X security focuses on eliminating misbehaving vehicles by detecting malicious Vehicle-to-Vehicle (V2V)~\cite{so2018integrating,van2018survey}, V2I~\cite{huang2021data}, or vehicle-to-cloud (V2C) messages~\cite{belhadi2020trajectory}. Additionally, any connected vehicle can be a "witness" and share misbehavior reports with the misbehavior authority (MA)~\cite{kamel2020simulation,mahmoudi2020towards}. The MA can then decide whether to revocate V2X certificates assigned to a given vehicle based on the misbehavior reports, as shown in Fig.~\ref{fig1} (the right part).

Although these approaches are proved to be effective in detecting and reporting malicious behaviors to the MA, there may not exist enough "witness" vehicles with V2X connectivity in the early stage of connected vehicle deployment~\cite{alnasser2019recommendation}. More importantly, owners of connected vehicles need incentives to opt-in for V2X services and share their observations. In general, previous work does not answer another fundamental question regarding the adoption of V2I services: How can we incorporate incentive mechanisms into V2X protocols to encourage each vehicle to share its own data about vehicle movement, traffic events, and trust evaluations for other vehicles?

As a complement to the previous "blacklisting approaches" that pay more attention to eliminating misbehavior by detection, this paper proposes a Proof-of-Travel (POT) protocol to simultaneously tackle the issues of incentive designs and the trustworthiness of V2I-reported data. The POT protocol is essentially a "whitelisting approach" in the sense that it focuses on "impeding" misbehavior by increasing the cost of being malicious, as shown in Fig.~\ref{fig1}. At the same time, the POT protocol can incentivize connected vehicle adoption and V2I data exchange by enabling normal vehicles to gain a reputation from the travel proof. The desire to maintain a high reputation score, which is derived from cryptography mechanisms embedded within the POT protocol, can then serve as the incentive to promote the sharing activities of V2I data. While building a chain of proofs of spatial movement creates burdens for malicious vehicles to launch attacks, acquiring proofs does not result in extra costs for normal vehicles, which naturally want to move from the origin to the destination. Compared to our previous work in~\cite{suo2020proof}, this paper makes extensions and contributions in the following aspects.

\begin{itemize}
    \item Setting up the trust foundation for incentivizing V2X adoption. Previously, participants for V2X testing activities are often chosen on a voluntary basis. For example, a traditional vehicle model whose owner wants to opt in must be installed with aftermarket onboard units (OBUs) supporting wireless vehicular communication~\cite{TampaParticipant2018}. Therefore, the automotive and transportation industries need to find effective and secure ways of engaging more stakeholders and customers before V2X technologies become ready for mass deployment. Inspired by the altruistic behaviors of users to gain reputational points in crowdsensing platforms, we propose a proof of travel protocol (POT) to measure each vehicle's contributions to vehicular networks in a trustworthy manner. The proposed POT protocol transforms a vehicle's movement into reputational points by using cryptographic techniques such that it is impossible for an adversary group to gain a reputation unless they are willing to pay the extra cost incurred by "compulsory" spatial movement. Unlike previous proof-based consensus protocols~\cite{liu2018blockchain,liu2019blockchain}, which introduce computational burden for both normal and malicious nodes, the voting algorithm building on the POT protocol aims to prevent malicious nodes from gaining disproportional influence on the network by increasing their burdens of spatial movement without incurring extra cost for normal vehicles. As a result, adversaries will lose interest in tampering with or fabricating V2I data if the cost of these malicious behaviors is greater than the rewards they earn, which is justified by using probabilistic analysis.
    \item Improving previous voting-based data validation algorithms by using the proposed POT protocol. While both the traditional plurality voting and the POT-enhanced voting algorithms can be used to verify the authenticity and accuracy of traffic events shared by nearby vehicles in real time, the latter is more robust. Specifically, the POT-enhanced voting prevents the trust mechanisms from being abused by adversaries with forged, incorrect, or inaccurate information about traffic and roads~\cite{rostamzadeh2015context}. With reasonable assumptions, we prove that all vehicles choosing to cheat is not a pure Nash equilibrium based on game-theoretic analysis.
    \item Evaluating the security and the performance of the POT-enhanced data validation algorithm through a simulation study. For security, the simulation results suggest that the proposed POT-enhanced voting algorithm is more robust to malicious V2I reports than the benchmark algorithm. When the POT protocol is used, at least a 10$\%$ decrease in the percentage of invalid events can be achieved with reasonable voting parameters. Additionally, the simulation results support the tradeoff decision between the event criticality (how much tolerance we have regarding event correctness) and the timing constraint (how quickly the event needs to be confirmed and routed) for different types of V2I-based crowdsensing applications.
\end{itemize}

This paper is organized as follows. Section II presents previous work on the incentive design for V2X communication and the security mechanisms for validating V2I data. In Section III, we present assumptions about adversaries and the vehicular network. The detailed procedures of collecting proof of travel by vehicles and V2I voting game for traffic event validation are given in Sections IV and V. In Sections VI and VII, we perform security analysis on the proposed POT protocol and the POT-enhanced voting algorithm and evaluate their security and performance through a simulation study. Section VIII discusses issues regarding real-world deployment. Section IX concludes by summarizing the main results and presenting unresolved issues in the POT protocol as future work.

\section{Previous work}
Previous work on the incentive designs and the trustworthiness of V2I-reported data disseminated in vehicular networks are often closely related. Specifically, there exist security mechanisms designed for two stages in the V2I data life cycle: the process of data transmission between vehicles and infrastructure and the process of data recording in distributed ledgers maintained and shared among infrastructure components, as summarized in Table I.

\begin{table*}[!htb]
\centering
	\caption{Previous work on the incentive designs and the trustworthiness of V2I data dissemination and recording}
		\begin{tabularx}{\textwidth}{|X|X|X|X|X|X|X|}
		    \hline
            &\multicolumn{3}{c}{\textbf{Data transmission between vehicles and RSUs}}&\multicolumn{3}{c}{\textbf{Data recording in distributed ledgers shared among RSUs}} \\
			\hline
			\textbf{Authors}&\textbf{Incentives for vehicles sharing data}&\textbf{Trust on the vehicle itself}&\textbf{Authenticity of vehicle-shared data}&\textbf{Incentives for RSUs recording data}&\textbf{Trust on RSUs}&\textbf{Consistency and Tamperproofing of data records} \\
			\hline
			\hline
			Liu et al. ~\cite{liu2018blockchain}&Data coin based on data contribution frequency&&Verified by RSUs, no details&&& POW, RSUs winning in POW as validators\\
			\hline
			Kang et al. ~\cite{kang2019toward}&&&Verified by RSUs, no details&Encourage high-reputation miners, i.e., RSUs, to participate based on contract theory&Aggregation of trust ratings by nearby vehicles&Delegated POS, RSUs with higher reputation as validators\\
			\hline
			Li et al. ~\cite{li2018creditcoin}&Data coins by forwarding V2X message&The amount of coins owned&Voting algorithms with Byzantine fault tolerance&&&\\
			\hline
			Yang et al. ~\cite{yang2018blockchain}&&Aggregation of trust ratings by nearby vehicles&Voting algorithms based on Bayesian inference&&&POW and POS, RSUs with higher stake more likely to win in POW\\
			\hline
			Liu et al. ~\cite{liu2019blockchain}&&Aggregation of trust ratings by nearby vehicles and false message rate&Identity-based group signature&&&POW, RSUs winning in POW as validators\\
			\hline
			Chang et al. ~\cite{chang2011footprint}&&Path of physical movement attested by timestamp Location certificates issued by RSUs&&&Assuming that RSUs are fully trustworthy&\\
			\hline
			Park et al. ~\cite{park2013defense}&&&Timestamp location certificates issued by RSUs&&Assuming that RSUs are fully trustworthy&\\
			\hline
			Liu et al. ~\cite{liu2019byzantine}&&Proof of knowledge shown by vehicles&Voting algorithms with Byzantine fault tolerance&&&\\
			\hline
			Baza et al. ~\cite{baza2020detecting}&&Proof of work shown by vehicles&Path of physical movement attested by timestamp Location certificates issued by RSUs&&Relying on threshold signature (from multiple RSUs) as one RSU can be compromised& \\
			\hline
                Falco et al. ~\cite{falco2020distributedblack}&&Distributed hash table&Redundency in the distributed hash table&OEMs to protect data and secure OTA update &Internal server owned by OEMs&POW, full blockchain node run by OEMs or fleet manager \\
                \hline
                Chan et al. ~\cite{chan2021towards}&&In-vehicle network blockchain&Hash record stored on blockchain&&& \\
                \hline
			\textbf{This paper}&Reputation or stake earned by proof of vehicle movement &Path of physical movement attested by timestamp Location certificates issued by RSUs&Voting algorithms in which eligibility to vote is determined by proof of travel&&RSU semi-trusted&\\
			\hline
		 \end{tabularx}
		\label{tab1}
\end{table*}

For the process of data transmission between vehicles and infrastructure, previous work has discussed incentivizing each connected vehicle to share the data regarding its status, movement, and observations of the surrounding environment by digital rewards (e.g., digital coins)~\cite{liu2018blockchain,li2018creditcoin}. However, rather than determining the reward based on the data transmission frequency\cite{liu2018blockchain} or the number of times a vehicle participates in data transmission and forwarding activities\cite{li2018creditcoin}, our paper proposes the idea of using a vehicle's distance traveled to measure its contributions and thus determine its rewards. One unique feature of the distance value in the proposed POT protocol is that the vehicle miles traveled is verifiable as the corresponding distance is derived from a chain of location proofs the vehicle has acquired from RSUs along its path of movement~\cite{falco2020distributedblack}. 

After an RSU receives data shared by a vehicle, the RSU needs to verify the digital identity of the vehicle for Security considerations. In addition to authentication based on cryptography methods, such as message authentication code and digital signatures~\cite{zhang2012vehicle}, there are two trust-based methods for identity verification based on vehicle reputation. First, the digital rewards earned by the vehicle can be used to determine its reputation as the value measures the contribution the vehicle makes to data-sharing activities after it joins the network for a period of time. This type of method can suffer the potential of Sybil attacks in which an inside adversary use stole vehicles or compromise vehicle credentials to fabricate multiple digital identities such that the adversary can gain a disproportionately stake or influence on the network~\cite{chaubey2020taxonomy}. One solution to prevent identity spoofing is to require the vehicle to show location proofs, which are issued by RSUs the vehicle has interacted with during the path of its movement~\cite{chang2011footprint,park2013defense,baza2020detecting}. The second way is to have a trusted infrastructure aggregate rating of trust evaluation from nearby vehicles with which the target vehicle has exchanged data to derive a total reputation score~\cite{yang2018blockchain,liu2019blockchain}. This type of method assumes a certain level of the penetration rate of vehicles with V2X connectivity.

If vehicle-shared data contains information about a traffic event with a high criticality level (e.g., incidents) and requires emergency responses, an RSU needs to verify the authenticity of these data before disseminating and reporting them to the transportation management center. In addition to relying on cryptography-based approaches~\cite{liu2018blockchain,kang2019toward,liu2019blockchain} or redundancy by relying on heterogeneous channels~\cite{falco2020distributedblack}, the RSU can also utilize voting algorithms that are Byzantine fault tolerant to determine the event accuracy and authenticity when it receives multiple V2I messages regarding the same traffic event~\cite{liu2019byzantine}. Our paper develops a voting algorithm in which the weight of each vote corresponds to the sender's reputation score determined by its previous proof of travel behaviors, which is different from previous voting-based approaches for V2I message authenticity~\cite{li2018creditcoin,yang2018blockchain,liu2019byzantine}.

Although the process of V2I data recording (in Table I) is not the focus of this paper, we summarize previous work to give readers a complete view of the V2I data life cycle. In particular, recent work recommends the use of distributed ledgers shared among RSUs to store V2I-shared data for future use, such as post-event investigation by law enforcement agencies~\cite{cebe2018block4forensic,guo2018blockchain}. For this reason, we need to have mechanisms for ensuring data consistency between different ledgers (i.e., sequence of data records) and preventing compromised RSUs from adding fabricated data records or tampering with the ledgers. Previous work relies on the proof-of-work (POW)~\cite{liu2018blockchain,liu2019blockchain}, proof-of-stake (POS) protocols~\cite{kang2019toward}, or the combination of both~\cite{yang2018blockchain} to select validators among RSUs, which construct, audit, and propose new data blocks to be written into distributed ledgers. It is worth mentioning that the issues of RSUs' incentives of contributing their computation and storage resources for data recording~\cite{kang2019toward} and the determination of trust in RSUs have not been fully explored and deserve more attention.

\begin{figure*}[tb!]
\centering
\includegraphics[width=0.8\textwidth]{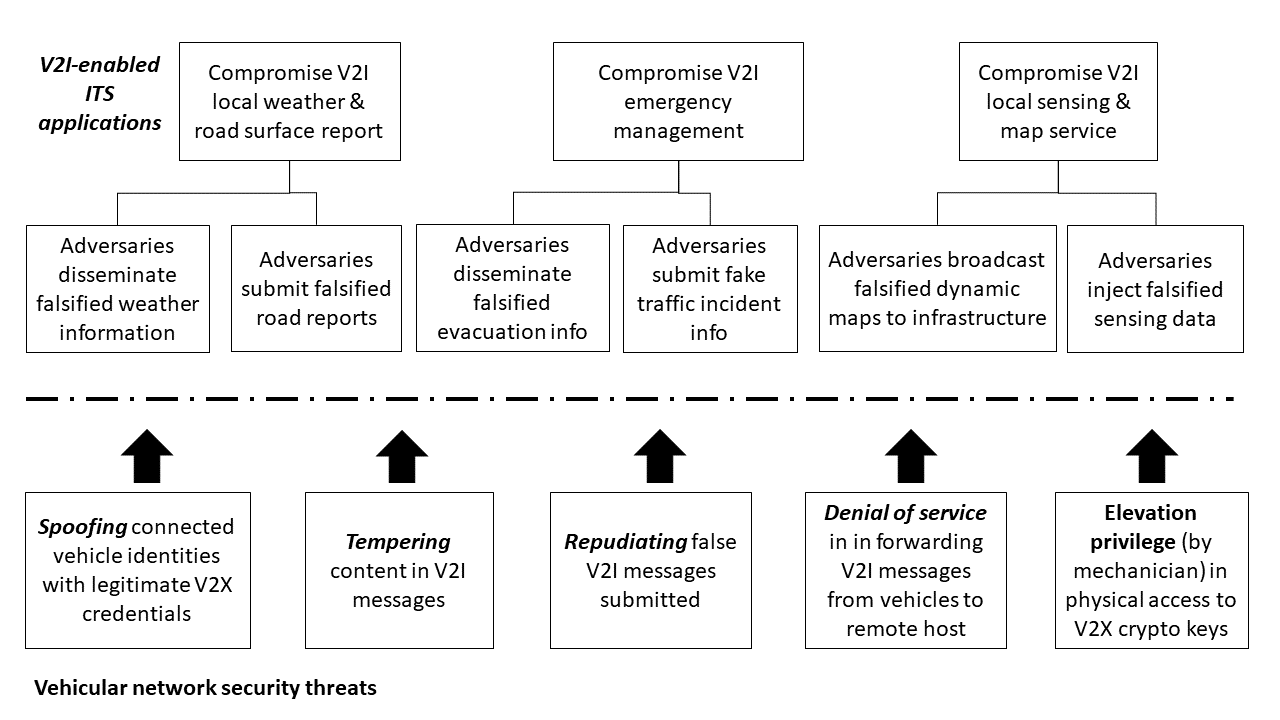}
\caption{The attack model of V2I-enabled crowdsensing applications.}
\label{attack_model}
\end{figure*}

\section{Attack models and system assumptions}

The POT protocol targets adversaries who want to compromise V2I-based transportation applications that utilize crowdsourced information from multiple vehicles. In particular, we are interested in the inside adversary who can indirectly compromise the V2I services by spoofing fake vehicle identities and submitting falsified events or sensing data. Also, the adversary can spoof multiple digital identities to flood the RSUs with bogus information. These attacks become possible if someone can get easy access to cryptography materials embedded within the onboard units (OBUs), such as "evil mechanics"~\cite{monteuuis2018attacker}. Due to the existence of after-market devices during the early phase of V2X deployment, it is straightforward for someone to gain expertise in manipulating V2X OBUs~\cite{suo2020location}. For example, the local transportation agency in Tampa, Florida in the U.S. was trying to hire instructors and students from a local college through a paid internship to install over 1,000 OBUs in vehicles that participated in the connected vehicle pilot program~\cite{tempa}.

Additionally, we assume that adversaries are motivated by profit-seeking behaviors, i.e., adversaries can gain economic benefits from fabricating vehicle identities and spoofing V2I data with false traffic events. We assume that adversaries are rational in an economic sense: adversaries will give up launching attacks if the rewards of being malicious are less than the cost incurred by following the POT protocol. From the perspective of game theory, each vehicle, whether it is normal or malicious, always tries to maximize its payoff. 

This is because the proposed POT protocol is a "whitelisting" approach that impedes adversaries by increasing the cost of launching attacks, rather than trying to detect malicious behaviors or eliminating adversaries. For this reason, the only prerequisite for each adversary to launch the attacks discussed earlier is that the attacker must be in close proximity to the target RSUs (i.e., within RSUs' communication range) to acquire location signatures (a concept that will be presented later) in a sequential manner. 

To further illustrate the relations between the attacks and transportation applications, we present attack trees for V2I-based transportation services that rely on crowdsource sensing information shared by CVs, including road surface or local weather notifications, traffic event and emergency management and local sensing and dynamic map services, as shown in Fig.~\ref{attack_model}. They are derived based on the architecture designs for the communication between connected vehicles and the roadside computing platforms proposed by the U.S. Department of Transportation (DoT)~\cite{chang2017overview}. 

\paragraph{Weather and road surface conditions report through crowd-sensing}
Since weather and road surface conditions can influence traveler’s decisions, pavement friction, and vehicle performance and stability, it is beneficial for both transportation agencies and individual travelers to get precise and dynamic views of road surface and environmental conditions along all road segments~\cite{chang2015estimated}. Other than relying on fixed-position roadside infrastructure-mounted sensors to monitor road weather conditions, transportation agencies may also leverage information from mobile sources. In regions where infrastructure-mounted sensing devices only provide limited coverage of observed areas, CV-based crowdsensing can provide timely information about the weather and road surface conditions that change dynamically. Previously with smartphones, citizens could download apps to report real-time weather and road surface conditions to transportation agencies or weather service stations~\cite{gopalakrishna2013utah}. 

In the near future, as the adoption of CVs increases,  we can leverage V2I communication to use CVs as sensors for Spot Weather Information Warning (SWIW) application~\cite{hadi2019connected}. One example is the such as the Weather Data Environment program (WxDE) in which environmental data collected by CVs can be shared with the transportation management center through V2I channels~\cite{WxDE}. There exist industry standards that define operating principles for V2I-based weather applications. For example, the society of automotive engineers (SAE) has published J2945/3, which defines interface requirements between vehicles and infrastructure for weather applications~\cite{SAE_J2945}. 

While providing assistance to travel decisions or transportation management, the services can be compromised if the information shared by CVs is falsified or intentionally fabricated. Although the security vulnerability might not directly lead to safety incidents, the inconvenience due to mistakenly or even maliciously reported weather and road information can result in customer dissatisfaction and destroy people’s confidence in connected vehicle technologies.

\paragraph{Transportation emergency management}
Similar to weather and road services, transportation agencies can use emergency event information shared by CVs to make decisions on sending professional teams for rescuing and medical services. More importantly, the information is crucial for evacuation plans in the events of natural or man-made disasters~\cite{hadi2019connected}. For example, local DoTs in the U.S. have been exploring the use of connected vehicle data to enhance their situational awareness of real-time traffic and the actual path of hurricanes such that citizens can be notified during the evacuation process~\cite{cv_hurricane1}. Failures in effective planning or selecting efficient routes due to erroneous or even malicious information can result in delays in evacuation or loss of life. 

\paragraph{Local sensing and dynamic map services} An important V2I application is to fuse sensing information from different traffic participants to form a “bird’s eye view” of safety-critical areas, such as intersections. For example, the standard on the data format and communication protocols for sharing and fusing individual sensor information for building Local Dynamic Map (LDM) services were discussed based on ad hoc vehicular networks~\cite{ETSI_LDM}. Vehicles can then rely on the local dynamic map (LDM) services provided by RSUs for collective perception and detecting objects in blind spots~\cite{ETSI_CP}. In the future, as the adoption of 5G continues to grow, we can also leverage the high bandwidth enabled by millimeter-wave communication for sharing raw data between vehicles and infrastructure. Vehicles that are installed with high-definition (HD) sensors, such as LiDAR, can upload raw sensing data to edge servers~\cite{5gaa_v2x}. Network operators can then leverage such crowd-sourced sensing information to build and update HD maps in real time. These efforts on standardization and cyber-infrastructure development open new doors to building new applications for vehicle perception and traffic control~\cite{3gpp_v2x}. 

Although the timing constraints for building HD maps need not be enforced with hard timing constraints as we did for collision avoidance applications, falsified crowd-sourcing data can compromise map quality and may eventually raise safety concerns. While unauthorized sensing data contributors can be eliminated by traditional security mechanisms (e.g., a firewall) installed on the edge servers, it is difficult to filter out malicious sensing data from vehicles that hold valid V2X credentials. An inside adversary may tamper with HD sensor data or inject noises to transform information attached to the shared data. 

\begin{figure*}[tb!]
\centering
\includegraphics[width=0.9\textwidth]{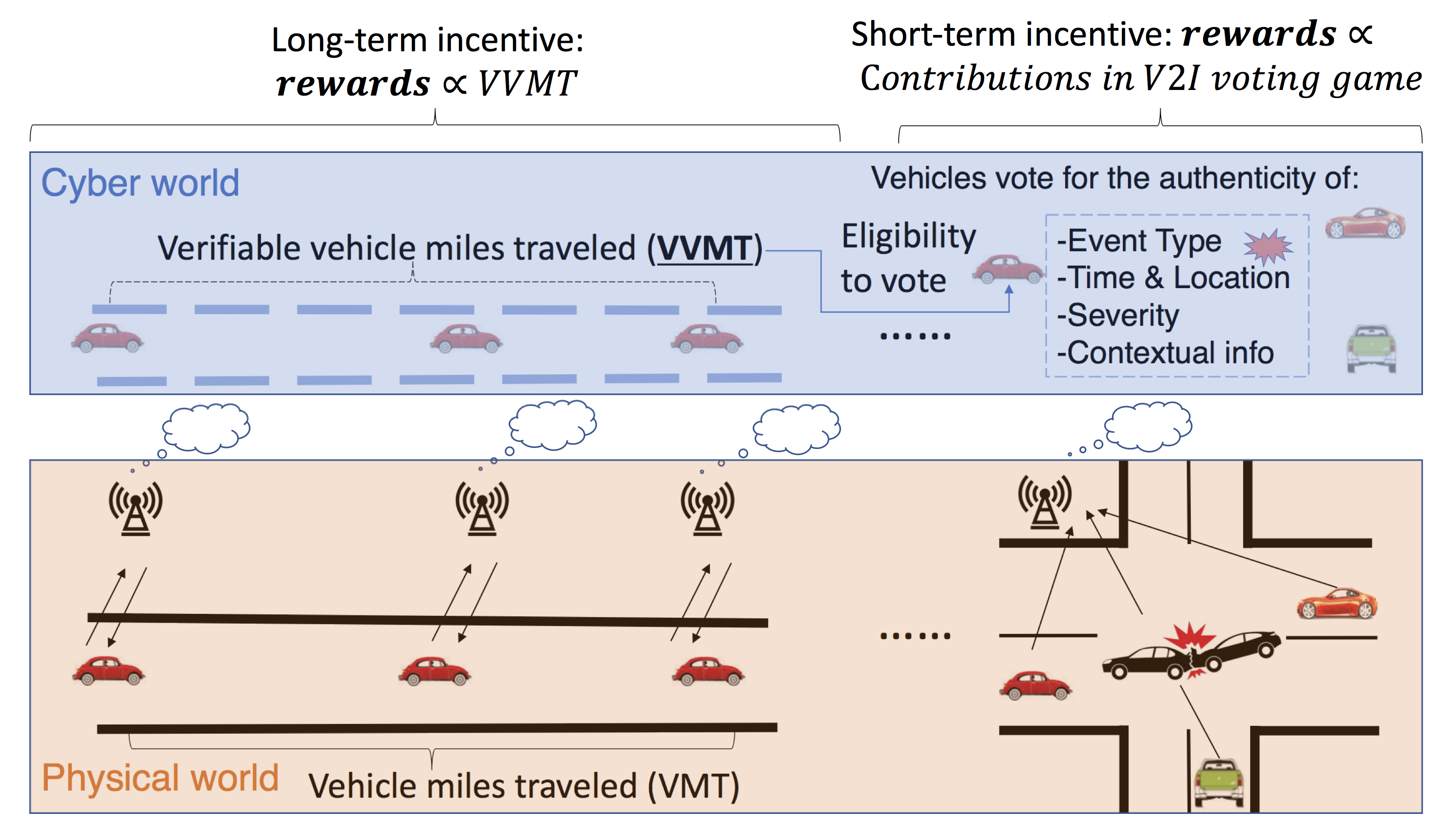}
\caption{The proposed long-term and short-term incentive designs for trustworthy V2I communication.}
\label{fig2}
\end{figure*}

\section{A Proof of Travel protocol}
\subsection{Overview of the incentive designs}
The mechanisms presented in this paper are motivated by incentive designs in smartphone-based crowdsource sensing platforms~\cite{restuccia2016incentive} in digital social networks, such as waze, a mobile app that ranks users based on digital points each one owns~\cite{Wazeranking}. Previous studies show that reputational rewards, the desire to be regarded as respectable due to altruistic behavior~\cite{ariely2009doing}, can motivate user participation in content-sharing activities. This implies that it is possible to leverage the desire of vehicle owners to gain a reputation for encouraging them to opt-in to V2I services.

Similar to social networks, each entity in vehicular networks may have the incentive to build reputation, which determines its power and stake in the network in the long run~\cite{yu2019repucoin}. This section discusses how we can utilize cryptography techniques, such as public-key infrastructure, hash functions, and digital signatures, to build communication protocols between vehicles and infrastructure to allow vehicles to build a chain of travel poofs issued by infrastructure components. The accumulated proofs assigned to a vehicle represent its contributions since it joins the network and can be used to measure its reputation.


The proposed protocol that enables the proof-accumulation process between vehicles and infrastructure aims to encourage each vehicle to share V2I data about its movement and status (e.g., position, speed, and acceleration, etc.) along its path of movement by rewarding each vehicle with digital points, as shown on the left in Fig.~\ref{fig2}. To hold more “stake” and earn reputational points in the vehicular network, a vehicle must construct a chain of proofs for its spatial movement. Specifically, the vehicle may request proof of its presence in a given location at a particular time by sending its movement status and credential information whenever it meets a roadside unit (RSU). Since each proof is linked to its previous one through the hash function, a type of cryptography mechanism, these concatenated proofs become tamper resistant and form a chain of trust. While building a chain of proofs of spatial movement creates burdens for malicious vehicles whose only goal is to compromise the system, acquiring proofs does not result in extra cost for normal vehicles, which naturally want to move from the origin to the destination. It is this asymmetric cost in traveling that helps the proposed POT protocol to mitigate malicious vehicles that try to gain adversarial rewards.   

These reputational points that each vehicle gained from travel proof can be used for data validation since an RSU may aggregate information about a traffic event from different vehicles in the surrounding region before disseminating the event to the TMC, as shown on the right in Fig.~\ref{fig2}. Compared to previous work on plurality voting-based methods for achieving Byzantine fault tolerance in distributed systems~\cite{liu2019byzantine}, the security of the voting game we proposed in this paper is strengthened as only vehicles which have gained a pre-determined amount of reputational points are eligible to vote. This is crucial for efficient transportation emergency management as the TMC must have accurate and correct data about traffic events for allocating rescuing and other necessary resources to event sites.

\subsection{Preliminaries}
POT defines the V2I message format and the communication procedure for a vehicle to acquire location proofs, titled location signature (formally defined later), from each RSU along its path of movement. The chain of proofs held by a vehicle testify both the vehicle's claimed trajectory and its contributions, forming the foundations for building incentive mechanisms for V2X services. We give core definitions in the POT protocol before presenting detailed communication procedures. 

\begin{table}[t]
	\caption{Notation used in the proof collection}
		\begin{tabularx}{\columnwidth}{c|X}
			\hline
			\textbf{Symbol}&\textbf{Meaning} \\
			\hline
			\hline
			$v_i$&A given vehicle denoted as i in vehicular networks  \\
			\hline
			$s^i_t$&The status of vehicle i at time t \\
			\hline
			$pk_{v_i},sk_{v_i}$&public and private key pairs held by vehicle i \\
			\hline
			$\sigma_{v_i}$&Digital signature of vehicle i\\
			\hline
			$\sigma_{rsu_j}$&Digital signature of rsu j\\
                \hline
			$pk_{rsu_j},sk_{rsu_j}$&public and private key pairs held by RSU j \\
			\hline
			$h_e$&The hash value of the information ($e^i_t$) reported by vehicle i \\
			\hline
			$ls^{<t>}_{i,j}$&Location signature issued by RSU j to vehicle i at time t \\
			\hline
			$h_{pre}$& The hash of the location signature $v_i$ acquired from the previous RSU along its path of movement \\
			\hline
		 \end{tabularx}
		\label{tab:POT_procedure}
\end{table}

\textit{Definition 5.1} A \textit{location signature} (denoted as $ls^{<t>}_{i,j}$) issued by a RSU (denoted as $rsu_j$) to vehicle $v_i$ at time $t=t_k$ is defined as 
\begin{center}
$pk_{rsu_j}||pk_{v_i}||t_k||h_e||h_{pre}||\sigma_{rsu_j} (pk_{rsu_j}||pk_{v_i}||t_k||h_e||h_{pre})$, 
\end{center}

After receiving vehicle i's request for location signature, $rsu_j$ will check i's digital signature $\sigma_{v_i}$ and sign on the content within the request to generate the RSU's corresponding signature $\sigma^{t_k}_{rsu_j}$, which testifies the vehicle i's presence in a particular location at a given time.

The definition presented here is different from the previous work~\cite{chang2011footprint,baza2020detecting} researching vehicle location as proof. In particular, the contents signed by the RSU contain the hash of the previous location signature $v_i$ collects earlier, and the hash of vehicle reported events $h_e$. The former ensures any changes in a location signature will result in the inconsistency between the hash of that particular signature and the pre-hash value in the next signature in the same chain of proofs, preventing trajectory forgery. The hash value of events $h_e$ plays the role of “admitting” the contribution made by $v_i$ as the shared events can be used by the TMC and other vehicles, as mentioned earlier. 

\textit{Definition 5.2} \textit{Proof of Travel} for vehicle $v_i$ is the chain of location signatures = {$ls^{t_0}_{i,j},ls^{t_1}_{i,j+1},...,ls^{t_T}_{i,j+T}$} that $v_i$ acquired from RSUs along the path of its movement during the time interval $T$.

\subsection{Detailed Communication Procedure}
The POT protocol consists of three stages. A vehicle will start to acquire and accumulate proofs from RSUs after it joins the vehicular network, and the process continues until it exits the road with V2X coverage. The reputation and reward of each vehicle participating in V2I communication activities (indicated by VVMT) are then determined by the length of the chain of valid proofs the vehicle collects.

\subsubsection{Stage-1: Initial proof generation}
The proof-collection process begins with a vehicle ($v_i$) sending a location-signature request $req^{t_1}_{v_i} = pk_{v_i}||t_1||e_i^{t_{1}}||pos_i^{t_{1}}||\sigma_{v_i} (pk_{v_i}||t_1||e_i^{t_{1}}||pos_i^{t_{1}})$ to the first RSU ($rsu_1$) it meets after joining the vehicular network, as shown in Fig.~\ref{fig3}. The request consists of the identity and authentication information of $v_i$, such as its encoded public key $pk_{v_{i}}$ and the digital signature signed on this request $\sigma_{v_i} (.)$, and the observed traffic events $e_i^{t_{1}}$ or its own real-time position $pos_i^{t_{1}}$.

After receiving the request, $rsu_1$ will authenticate $v_i$'s identity and the integrity of the request message by using the attached digital signature $\sigma_{v_i} (.)$. It may also check the plausibility of the information included in the request, such as $e_i^{t_{1}}$ or $pos_i^{t_{1}}$, by using pre-defined heuristic rules. For example, if the location of the event reported by the vehicle is far away from the vehicle’s own location, or the vehicle is impossible to pass the event site along its path, $rsu_1$ may reject the request.

If results from all the identity, integrity, and rule-based plausibility checks turn out to be valid, $rsu_1$ will generate a location signature $ls^{t1}_{v_i,rsu_1}$ and send it back to $v_i$. 

\begin{figure*}[tb!]
\centering
\includegraphics[width=0.7\textwidth]{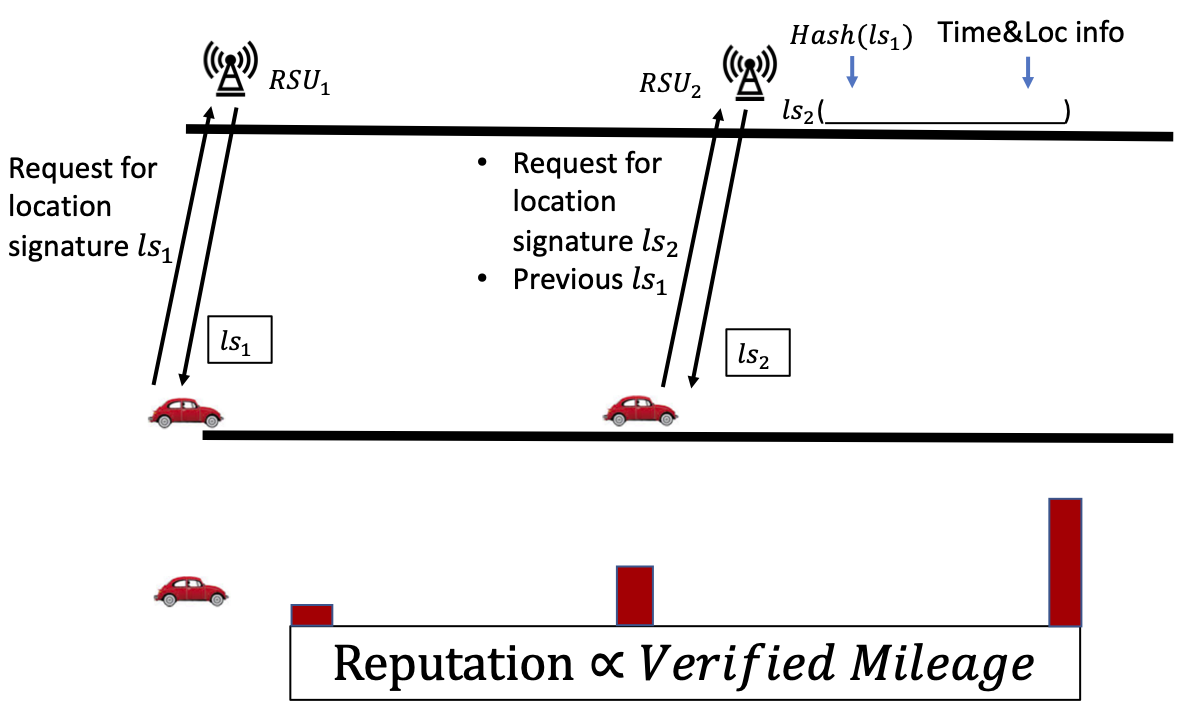}
\caption{The process of building Proof of Travel by vehicles.}
\label{fig3}
\end{figure*}

\subsubsection{State-2: Trajectory-encoded proof collection}
When vehicle $v_i$ meets the next RSU $rsu_{2}$, it will attach the location signature $ls^{t1}_{v_i,rsu_1}$ acquired from $rsu_{1}$ when sending a new request $req^{t_2}_{v_i}$, as shown in Fig.~\ref{fig3}. Similarly, in addition to verifying the new request based on the heuristics described earlier, $rsu_{2}$ will also check if the previous location signature $ls^{t1}_{v_i,rsu_1}$ is owned by $v_i$, the vehicle sending the request (ownership checks), has not expired (time of validity), and is issued by a legitimate RSU (legitimacy checks), such as a valid RSU who has registered with the trust authority and is adjacent to or near $rsu_{2}$. 

If all checks are valid, $rsu_{2}$ will construct a new location signature $ls^{t2}_{v_i,rsu_2}$. However, other than concatenating all the elements as discussed earlier, $rsu_{2}$ will also attach the hash value of the previous location signature $ls^{t1}_{v_i,rsu_1}$, sign on the merged data, and send it back to $v_i$. This process repeats until the vehicle has collected enough location signatures to form a chain of proofs, which are geo-time-stamped ledgers of $v_i$'s trajectory history verified by all RSUs along its path of movement.

\subsubsection{Stage-3: Determine vehicle stake or reputation based on POT}
The TA or any RSU can rely on the chain of proofs owned by $v_i$ to determine its reputation (and the reward). We propose verifiable vehicle miles traveled (VVMT), a concept derived from the proofs of location signatures a vehicle holds, to quantify the reputation and the stake the vehicle has in the network.

VVMT, which incentivizes V2I data sharing by quantifying a vehicle's contributions to improving traffic conditions, is quite different from vehicle miles travel (VMT), which is a concept used extensively in transportation planning, such as travel demand forecasting and dynamic pricing for toll lanes~\cite{williams2016methodologies}. For example, a vehicle traveling more distance on a congested road during peak hours will be charged higher tolling fees because the vehicle is regarded as contributing to the congestion. In this regard, VMT plays the role of a “negative” incentive to discourage vehicles from using a road over a particular time period, while VVMT serves as a positive incentive to encourage the use of connected roads with V2X connectivity.

\textit{Definition 5.3} \textit{Verifiable vehicle miles traveled} (VVMT) for a vehicle, denoted as $v_i$, is defined as a function of all proofs of location signatures $v_i$ has collected since it joined the network, as given in eq.~\ref{eq:vvmt_equ}

\begin{equation}\label{eq:vvmt_equ}
	vvmt^{T}_{i} = f(LS^{T}_i,n)
\end{equation}

Based on \textit{Definition 5.3}, we present two equations for deriving the vehicle's corresponding reputation scores, including the \textit{vanilla VVMT} (eq.~\ref{eq:vanilla_VVMT}) and the \textit{resilient VVMT} (eq.~\ref{eq:flexible_vvmt}) equations.

\paragraph{Vanilla VVMT}

\begin{subequations}\label{eq:vanilla_VVMT}
\begin{align}
vvmt^{V}_{i} &= \sum_{t=1}^T \gamma d(ls_i^{t-1},ls_i^{t}) \label{eq:vanilla_VVMTa}\\
vvmt^{V}_{i} &= \frac{M}{1+e^{-k(\sum_{t=1}^T d(ls_i^{t-1},ls_i^{t})-m)}}   \label{eq:vanilla_VVMTb}
\end{align}
\end{subequations}

In the \textit{vanilla VVMT} equation in~\ref{eq:vanilla_VVMT}, a vehicle's reputation score only depends on the physical distance the vehicle has traveled. The rationale behind this design is that since a vehicle can only accumulate the whole chain of location signatures by being in close proximity to every RSU along its path of movement, we may trust the vehicle by assigning the reputation score corresponding to its moving distance. 

Different functions can be used to realize the \textit{vanilla VVMT} for controlling how fast or difficult the vehicle can earn reputation points as it moves. Eq.~\ref{eq:vanilla_VVMTa} indicates a linear relationship between reputation and distance. $d(ls_i^{t-1},ls_i^{t})$ denotes the user-defined distance between two location signatures, such as euclidean distance or the actual distance $v_i$ has traveled. 

On the other hand, eq.~\ref{eq:vanilla_VVMTb} means that the vehicle's reputation scores grow exponentially. For the exponential form in eq.~\ref{eq:vanilla_VVMTb}, while $M$ is the scaling factor that determines the maximum reputation score a vehicle can earn, $k$ and $m$ serve the role of controlling the difficulty of gaining reputation. This design rationale is inspired by existing Blockchain-based cryptocurrencies in which each "miner" or contributor might have different difficulty levels when mining blocks to earn rewards.

\paragraph{Resilient VVMT}
\begin{subequations}\label{eq:flexible_vvmt}
\begin{align}
	vvmt^{R}_{i} = \alpha \sum_{t=1}^T \gamma d(ls_i^{t-1},ls_i^{t}) + (1-\alpha)\frac{n}{n_{max}} M \\
        vvmt^{R}_{i} = \alpha \frac{M}{1+e^{-k(\sum_{t=1}^T d(ls_i^{t-1},ls_i^{t})-m)}} + (1-\alpha)\frac{n}{n_{max}} M 
\end{align}
\end{subequations}

A prerequisite for applying the \textit{vanilla VVMT} equation is that the vehicle has to successfully collect all the location signatures from every RSU along its path of movement. However, this requirement can be difficult to meet due to communication failure or packet loss. To make the POT protocol more flexible and resilient by accounting for the possibility of missed location signatures during vehicle movement, we present the \textit{resilient VVMT}, as given in eq.~\ref{eq:flexible_vvmt}. Similar to the \textit{vanilla VVMT}, the \textit{resilient VVMT} is realized by both the linear and logistic versions.

When \textit{resilient VVMT} is applied, the value of $v_i$'s VVMT is determined by both the distance $v_i$ traveled and the ratio of the actual location signatures ($n$ in eq.~\ref{eq:flexible_vvmt}) the vehicle has collected to the total number of RSUs ($n_{max}$) along the road segment. The user-defined parameter $\alpha$ determines the ratio of contribution to $v_i$'s VVMT from each of the two aspects above. A case study on the VVMT accumulation in section~\ref{sec:exp} will be presented to illustrate the influence of different functions and parameters on the vehicle's reputation in the \textit{vanilla VVMT} equation.

It is worth mentioning that, when deriving the VVMT for a vehicle, the TA may also check if all location signatures included in a chain of proofs indicate a plausible trajectory. This type of plausibility check can build upon the "blacklisting approaches" for detecting V2X misbehavior and anomalous trajectory~\cite{van2018survey,huang2021data,belhadi2020trajectory,so2019physical}, which adds another layer of security defense to the trust management in the V2X ecosystem at the price of higher computation cost. For example, multiple location signatures owned by a vehicle may indicate an extremely fast speed impossible to achieve under the current traffic and road conditions. Also, they may form a strange trajectory (e.g., taking a zigzag line even if a straight line is an optimal route), which will reduce the likelihood of previous proofs being valid. 

The rules for verifying proofs should also support fault tolerance in the case where $v_i$ fails to get the location signatures from two adjacent RSUs due to faults or congestion in communication links, as we did in deriving vehicle reputation scores earlier. For example, a threshold signature scheme can be used in authenticating vehicle trajectory~\cite{baza2020detecting} such that only a subset (m) of all RSUs' signatures (n, $m<n$) is needed for determining the legitimacy of location proofs presented by a vehicle. 

With VVMT, engineers can develop reputation-based algorithms for validating vehicle-reported traffic events with high criticality levels. The next section will devote to the description of these algorithms and relates them to digital rewards gained by vehicles as short-term incentives. 

\section{POT-enhanced voting game for V2I validation}
\subsection{Why do we need voting-based methods?}
The proposed POT protocol in the last section allows a vehicle to gain a reputation as the incentive for V2I adoption in the long run. However, as a person who acts rationally, the vehicle owner may want to know how much digital rewards (s)he can get immediately after sharing data through V2I communication. More importantly, some traffic events included in these V2I data have high criticality levels, such as road incidents with serious injuries, and often require immediate emergency responses and medical services~\cite{shen2004managing}. Since these vehicle-reported data can help the TMC evaluate the severity and the surrounding environment of the incident site for effective emergency planning and operation, the accuracy and the correctness of these data need to be first validated by RSUs in the edge layer to avoid mistakes in assigning emergency response resources~\cite{chen2018every}. 

Built upon vehicle reputation derived from proof of travel, we propose a V2I voting game for validating vehicle-reported traffic events. Specifically, an RSU will verify the authenticity and the correctness of a given event based on V2I reports from multiple vehicles near the event site. Any vehicle which wants to be eligible for participating in the V2I game must have at least a predetermined level of reputation and thus accumulate enough proof of travel as vehicles with a higher score of VVMT will be viewed as more trustworthy (i.e., a higher reputation score). 

The design that uses the reputation indicated by VVMT of a vehicle to determine whether its event report is valid can improve the voting algorithms' tolerance for malicious vehicles~\cite{yu2019repucoin}. The reason is that it is extremely difficult for adversaries in the malicious group to acquire a high value of VVMT during a short period of time through physical movement. Even if a malicious vehicle can delegate the VVMT it has acquired to a colluding node to increase the weight the latter one has in a voting game, the malicious coalition as a whole still needs to pay the cost of spatial movement, making malicious behaviors of reporting falsified events less attractive from an economic perspective. Section VI will give a detailed discussion of this dilemma situation for the malicious group from a game-theoretic perspective.

We present two algorithms (Algorithm 1 and 2) for validating vehicle-reported events in the V2I voting game and evaluate the running time of the algorithms as there are often timing constraints for emergency response actions. However, when V2I-reported data about traffic events are only used for post-event investigation~\cite{cebe2018block4forensic,guo2018blockchain}, data accuracy and correctness, rather than timing constraints, become the highest priority. In the latter scenario, incident data must be accurate and consistent such that they can be used by law enforcement agencies. Depending on which goal the algorithms want to achieve, engineers can adjust the parameter that controls the minimum number of eligible votes required for an RSU to confirm a V2I-reported traffic event (i.e., $N_{thld}$ in eq.~\ref{voting_Rule}).  

\subsection{Voting rules}

\textbf{Vote}: A vote representing a vehicle’s opinions on an event at time $t$ is denoted as $x_{i,k}^{<t>}$, where $i=\{1,2,...,n\}$ represents the ID of the vehicle and $k= \{incident, workzone, congestion, ...\} $ represents the type of the reported event. This paper adopts a generic form for each vote sent by a vehicle and lets $x_{i,k}^{<t>}=\{-1,1\}$. For example, $x_{i,k}=1$ means that a vehicle has observed and reported the occurrence of a traffic event that requires immediate responses by the TMC. When deploying the POT-based voting protocol in various V2I applications, $x_{i,k}^{<t>}$ can have more complex forms or even be a continuous variable. One example is that the RSU aggregates V2I messages shared by multiple vehicles in a local area to estimate the starting location of congestion (SLoC) and travel time (TT) of the work zone before disseminating the information to the TMC and nearby vehicles in the same road~\cite{maitipe2012development}.

\textbf{Voting rule}: Let $X$ be the set of votes an RSU receives from $n$ vehicles in a local region, a voting rule is defined as a \textit{social choice function}~\cite{shoham2008multiagent,leonardos2020weighted} $X^n \rightarrow O$, which maps aggregated profiles of vehicle’s opinions $X^n$ to an outcome in $O$. The form of an outcome is consistent with that of the votes an RSU receives. For binary voting, $O$ has a binary form (i.e.,$O=\{-1,1\}$), indicating that the RSU makes decisions on whether to disseminate the event to TMC or nearby vehicles or not. Therefore, $O=1$ denotes that the RSU confirms the occurrence of the event and will forward it to the TMC to take action on it, and vice versa. Mathematically, our proposed voting rule can be described by eq.~\ref{voting_Rule}. 

\begin{equation}\label{voting_Rule}
  \setlength{\arraycolsep}{0pt}
  O = \left\{ \begin{array}{ l l }
    1, \quad if \underset{x_{i,k}==1}{\sum_1^{N_{thld}}} x_{i,k}^{<t>} \geq (2/3)*N_{thld}\\
    -1, \, otherwise
  \end{array} \right.
\end{equation}

$N_{thld}$ in eq.~\ref{voting_Rule} corresponds to the "committee size" in the context of consensus algorithms for distributed networking and represents the minimum number of eligible voters required in our discussion of V2I-data validation. For the conventional plurality voting (CPV) that will be presented as the benchmark algorithm (Algorithm 1), every vote from any vehicle is regarded eligible, which means that $N_{thld}$ equals the total number of votes received from vehicles. On the other hand, when the POT protocol is used for enhancing voting algorithms, a vote will be regarded eligible only if it is sent by a vehicle that has collected enough proof of travel (i.e., the vehicle's $VVMT$ greater than a predetermined threshold $VVMT_{thld}$). Therefore, under the proposed POT-enhance plurality voting (PPV), $N_{thld}$ can be less than the total number of votes received by the RSU since some votes without enough travel proof will be disregarded. 

A natural way to realize the voting rules given in eq.~\ref{voting_Rule} is by using conventional plurality voting (CPV)~\cite{liu2019byzantine,leonardos2020weighted} in which every vehicle in the communication range of the RSU is eligible to share data (i.e., vote), as shown in Algorithm 1. 

Alternatively, we can propose POT-enhance plurality voting (PPV) in which each vehicle's proof of travel is used to determine whether votes sent by the vehicle are regarded as eligible, as shown in Algorithm 2.

Here, we define the set of eligible voters for CPV and PPV as $\mathcal{I}^{\text{CPV}}$ and $\mathcal{I}^{\text{PPV}}$, respectively. Specifically,
\begin{align}
&\mathcal{I}^{\text{CPV}} = \{i:\text{Some location-based constraints}\} \\
&\mathcal{I}^{\text{PPV}}(VVMT_{thld}) = \{i:\text{Some location-based constraints},\notag\\
&VVMT_i > VVMT_{thld}\} 
\end{align}

In both the CPV and PPV algorithms, each vote sent by a vehicle has equal weight. For the CPV, this rule implies that the RSU will confirm that an event has occurred in a given location at a particular time as long as more than two-thirds of all nearby vehicles attest to it. For the PPV, since the RSU only recognizes a vote sent by a vehicle with enough proof of travel, this prerequisite for voting places the adversary group in a disadvantageous situation as it needs to make extra efforts to gain travel proof. We will provide formal proof to show that PPV indeed mitigates adversarial manipulation of votes from a game-theoretic framework in the next section.

\begin{algorithm}
\KwIn{$x_i^t$, $N_{thld}$, $N_{vo}, \mathcal{I}^{\text{CPV}}$}
\KwOut{$N_{vo}, \mathcal{I}^{\text{CPV}}$}
 \nl   initialize:  n $\leftarrow$ $N_{vo}$, H $\leftarrow$ $\mathcal{I}^{\text{CPV}}$\;
 \nl k $\leftarrow$ getEventType($x_i^t$)\;
 \nl \uIf{checkEventType(k) == False}{
    \Return }
 \nl $\sigma_{v_i}$ $\leftarrow$ getVehSig($x_i^t$)\;
 \nl \uIf{verifyVehSig($\sigma_{v_i}$) == False}{
    \Return }
 \nl n++\;
 \nl H.add($x_i^t$)\;
 \nl $x_{coalition} \leftarrow 0$\;
 \For {temp in H}{
    \nl   $x_{tmp} \leftarrow getVote(temp)$\;
    \nl  \If{$x_{tmp} == 1$}{ 
    \nl    $x_{coalition}$++}
  }
 \nl \uIf{$w_{coalition} >= \frac{2}{3}*n$ and $n \geq N_{thld}$}{
 \nl  RSU confirms and report the event to TMC\;
 \nl  $n$ = 0, H.clear()\;}
 \nl \uElseIf{$w_{coalition} < \frac{2}{3}*n$ and $n \geq N_{thld}$}{
 \nl  no consensus formed on the event\;
 \nl  $n$ = 0, H.clear()\;}
 \nl  $N_{vo}=n$, $\mathcal{I}^{\text{CPV}}$ $\leftarrow$ H\;
 \nl \Return{$n_{vote}, \mathcal{I}^{\text{CPV}}$}

 \caption{The conventional plurality voting (\textbf{CPV}) algorithm~\cite{liu2019byzantine,leonardos2020weighted} for verifying V2I-reported traffic events}
 
\end{algorithm}

\begin{algorithm}
\KwIn{$x_i^t$, $VVMT_{thld}$, $N_{thld}$, $N_{vo}, \mathcal{I}^{\text{PPV}}$}
\KwOut{$N_{vo}, \mathcal{I}^{\text{PPV}}$}
 \nl   initialize:  n $\leftarrow$ $N_{vo}$, H $\leftarrow$ $\mathcal{I}^{\text{PPV}}$\;
 \nl k $\leftarrow$ getEventType($x_i^t$)\;
 \nl \uIf{checkEventType(k) == False}{
    \Return }
 \nl $\sigma_{v_i}$ $\leftarrow$ getVehSig($x_i^t$)\;
 \nl \uIf{verifyVehSig($\sigma_{v_i}$) == False}{
    \Return }
 \nl $PoT_{v_i}$ $\leftarrow$ getTravelProof($x_i^t$)\;
 \nl \uIf{verifyVehProof($PoT_{v_i}$) == False}{
    \Return }
 \nl $VVMT_{v_i}$ $\leftarrow$ getVVMT($x_i^t$)\;
 \nl \uIf{$VVMT_{v_i} < VVMT_{thld}$}{
    \Return }
 \nl n++\;
 \nl H.add($x_i^t$)\;
 \nl $x_{coalition} \leftarrow 0$\;
 \For {temp in H}{
    \nl   $x_{tmp} \leftarrow getVote(temp)$\;
    \nl  \If{$x_{tmp} == 1$}{ 
    \nl    $x_{coalition}$++}
  }
 \nl \uIf{$w_{coalition} >= \frac{2}{3}*n$ and $n \geq N_{thld}$}{
 \nl  RSU confirms and report the event to TMC\;
 \nl $n$ = 0, H.clear()\;}
 \nl \uElseIf{$w_{coalition} < \frac{2}{3}*n$ and $n \geq N_{thld}$}{
 \nl  no consensus formed on the event\;
 \nl  $n$ = 0, H.clear()\;}
 \nl  $N_{vo}=n$, $\mathcal{I}^{\text{PPV}}$ $\leftarrow$ H\;
 \nl \Return{$N_{vo},\mathcal{I}^{\text{PPV}} $}

 \caption{The proposed POT-enhanced plurality voting (\textbf{PPV}) algorithm for verifying V2I-reported traffic events}
\end{algorithm}

\section{Security Analysis}
\subsection{Security of the POT protocol}
Since the communication between vehicles and road infrastructure in the POT protocol is vulnerable to eavesdropping, colluding, and insider attacks, we need to evaluate whether and how POT protocol can defend against different types of attacks on V2I channels.

\subsubsection{Replay attacks} 
Similar to other trajectory-based authentication approaches~\cite{chen2009robust,chang2011footprint}, the POT protocol can prevent the misuse of location signatures intercepted by an adversary when eavesdropping V2I channels. This is because the public key of a vehicle attached to the location signature can ensure that only the vehicle that holds the corresponding private key can claim the "ownership" of the location signature and use it as proof.

\subsubsection{Proof or trajectory forgery by "inside" adversaries}
Incorporating cryptography hashing (e.g., sha256) into the signed data can prevent the forgery of chains of proofs. Since we have added the hash of the previous location signature into the contents of the current location signature during stage 2 of the POT protocol, any changes made to a particular location signature will also change its hash value, which results in an inconsistency between the hash of that particular location signature and the pre-hash value contained in the next location signature in the chain of proofs. This also means that an adversary who holds valid vehicle credentials and wants to forge a valid chain of proofs containing multiple location signatures must be physically present in each corresponding RSU. The adversary must also follow a “plausible” trajectory to gain VVMT as any RSU only accepts a location signature request if the previous signature attached to it is signed by another legitimate RSU (legitimacy checks on location signature in Stage 2), and the trust authority will verify the plausibility of a vehicle's claimed trajectory indicated by its chain of proofs (trajectory plausibility checks in Stage-3). 

It is the cost of "compulsory" spatial movement for gaining location signatures that reduces the adversary’s incentive for being malicious. However, the requirement for spatial movement will not incur extra costs for normal travelers.

\subsubsection{Private key swapping by colluding nodes}
Although POT does not completely eliminate colluding nodes, it can diminish the negative effect of misusing location signatures for the same reasons discussed in proof-forgery attacks. For example, even if a malicious vehicle node may tunnel the location signature it collects along with its private key to another colluding node, the former one must also share a valid chain of proofs it acquired for the latter node to gain VVMT. Sharing only a subset of location signatures or modifying any location signatures will invalidate the whole chain of proofs. This also means that, from an economic perspective, the group of colluding nodes as a whole must always "pay the cost" incurred by spatial movement. When the cost becomes greater than the benefit the colluding nodes earn, they lose the incentive for forging V2I events. In the next section, we will formally prove this from a game-theoretic perspective.

\subsection{Security of POT-enhanced plurality voting algorithm}
Plurality voting game has been studied extensively in the game theory literature \cite{riker1968theory,palfrey1985voter,dhillon2004plurality, feddersen2006theory, myatt2007theory}. However, the plurality voting for verifying V2I-validation in this paper is different from the canonical voting game setting. First, in this paper, vehicles that do not vote will not get any reward, while in the canonical voting game, a participant, though not voting, can still get a reward if their preferred candidate is elected. Second, in this paper, we assume vehicles whose votes are in the minority group will get punishment due to their potential cheating behavior. Due to these different settings, we cannot directly quote results from the canonical plurality voting game.

In this paper, we propose a new framework to analyze the two voting schemes (conventional plurality voting (CPV) v.s. POT-enhanced plurality voting (PPV)). Results show that, \textbf{under the CPV, there is no specific way to eliminate the Nash Equilibrium (NE) that all vehicles cheat. However, under PPV, we prove that it is possible to eliminate this type of Nash Equilibrium (i.e., all vehicles cheat) by reasonably setting $VVMT_{thld}$ and $N_{thld}$.} The specific analysis is shown below.  

We make the following assumption about the driver's utility. The utility (payoff) matrix of a vehicle $i$ (i.e., $u_i(x_i,x_{-i})$) is shown in Table \ref{tab_payoff}.
\begin{table}[htb]
\centering
\caption{Payoff matrix of vehicle $i$ (i.e., $u_i(x_i,x_{-i})$)}\label{tab_payoff}
\resizebox{\linewidth}{!}
{
\begin{tabular}{@{}cccc@{}}
\toprule
Scenario                              & \multicolumn{1}{c}{$x_i = 1$} & \multicolumn{1}{c}{$x_i = -1$} &   \multicolumn{1}{c}{$x_i = 0$}\\ \midrule
$N_{vo} \geq N_{thld}$ and Majority$^1$ vote $x_i=1$  & $R-c$                                 & $-P-c-Vi$       &           0                 \\
$N_{vo} \geq N_{thld}$ and Majority vote $x_i=-1$ & $-P-c$                                  & $R-c-Vi$          &        0               \\
$N_{vo} \geq N_{thld}$ and No majority consensus$^2$ & $-c$                                  & $-c-Vi$          &        0               \\
$N_{vo} < N_{thld}$    & $-c$                                    & $-c-Vi$         &       0          
                                    \\ \bottomrule
\multicolumn{4}{l}{\begin{tabular}[c]{@{}l@{}}$^1$: Majority means more than $\frac{2}{3}$ of all voted vehicles. \\
$^2$:  No majority consensus means no voting group has more than $\frac{2}{3}$ voters. \\
\end{tabular}}
\end{tabular}
}
\end{table}
$R$ is the reward given to voted vehicles whose votes are in the majority group (i.e., more than $\frac{2}{3}$). $P$ is the punishment for cheating (i.e., voting in the minority group). $c$ is a fixed cost of voting, which accounts for the time and inconvenience of submitting responses to the system. $x_i = 1, -1, 0$ represents the vehicle $i$'s reporting truth, cheating, and not voting actions, respectively. $N_{vo}$ is the total number of voted vehicles and $N_{thld}$ is the threshold of number of voted vehicles to result in a valid outcome. We assume each vehicle $i$ has an integrity cost $V_i$ if they choose to cheat (i.e., $x_i = -1$). Note that $V_i < 0$ indicates malicious drivers who get positive payoff if cheating. We assume that $R>c$, otherwise no one will choose to vote. 

We consider a complete knowledge voting game where $V_i$ ($\forall i\in \mathcal{I}^{\text{CPV}}\; \text{or}\; \mathcal{I}^{\text{PPV}}$) is a common knowledge and known by everyone\footnote{It is also possible to consider a incomplete knowledge Bayesian game with $V_i$ as random variables. We leave this as a future study}. We define a driver $i$ as a ``super-integrity driver'' (SID) if $V_i > R-c$, which means for these drivers, their integrity cost is so high that they are not willing to win the vote by cheating. Define the number of SIDs in a group of driver $\mathcal{I}$ as
\begin{align}
    N_{SID}(\mathcal{I}) = \sum_{i\in\mathcal{I}} \mathbbm{1}[V_i > R-c]
\end{align}

\begin{assumption}\label{assump_int_cost}
In general, drivers with high $VVMT_i$ are more likely to be a SID. Mathematically, the increase in $VVMT_{thld}$ will increase the proportion of SIDs in $\mathcal{I}^{\text{PPV}}(VVMT_{thld})$ (denoted as $p_{SID}$), i.e., 
\begin{align}
    VVMT_{thld} \uparrow  \;	\Rightarrow \; p_{SID} := \frac{N_{SID}(\mathcal{I}^{\text{PPV}}(VVMT_{thld}))}{|\mathcal{I}^{\text{PPV}}(VVMT_{thld})|}\uparrow 
\end{align}
\end{assumption}

Based on Assumption \ref{assump_int_cost} (we justify this assumption in \textbf{Proposition}~\ref{prop_POT_secure} included in the appendix), we have the following main Proposition for the analysis of NE.
\begin{prop}\label{prop_NE}
In a plurality voting game (both CPV and PPV) with eligible voter set $\mathcal{I}$, all voting vehicles choosing to cheat ($x_i = -1$ for all $i$ choosing to vote) is \textbf{NOT} a pure NE if and only if  $p_{SID} > \frac{|\mathcal{I}| - N_{thld}}{|\mathcal{I}|}$. 
\end{prop}
\begin{proof}
\textbf{Sufficient condition}: if $p_{SID} > \frac{|\mathcal{I}| - N_{thld}}{|\mathcal{I}|}$, all voting vehicles choosing to cheat is not a pure NE.

Notice that for all SIDs, cheating (i.e., $x_i = -1$) is a weakly dominated strategy (by $x_i = 0$). Hence, we can eliminate $x_i = -1$ for all SIDs. As we assume the proportion of SIDs is greater than $\frac{|\mathcal{I}| - N_{thld}}{|\mathcal{I}|}$. The number of remaining drivers who may cheat is at most $N_{thld} - 1$. According to the payoff matrix (Table \ref{tab_payoff}), even if all these remaining drivers choose to cheat, the number of voters are less than $N_{thld}$, which will not result in an outcome of $O = -1$. Thus, their payoff will be $-c-V_i$. And anyone can switch to $x_i = 0$ to be better-off (with payoff equal to 0). Therefore, all voting vehicles choosing to cheat is not a pure NE. Under this scenario, the only pure NEs are 1) all eligible vehicles choose to report truth (i.e., $x_i = 1,\;\forall i \in \mathcal{I}$) with payoff $R-c$. 2) all eligible vehicles choose to not vote (with payoff 0)\footnote{We assume $N_{thld} > 1$, so anyone switch to vote will not change the outcome}.

\textbf{Necessary condition}: if $p_{SID} \leq \frac{|\mathcal{I}| - N_{thld}}{|\mathcal{I}|}$, all voting vehicles choosing to cheat is a pure NE.

The proof is similar, since now the number of non-SID drivers is at least $N_{thld}$. If all of them choose to cheat, everyone would has a positive payoff of $R-c-V_i$. And no one has incentive to switch to $x_i = 1$ (payoff $-P-c$) or $x_i=0$ (payoff 0). Under this scenario, there are three possible pure NE. The first two are the same as above. And the third pure NE is that all SIDs choose not to vote while all non-SIDs choose to cheat.  
\end{proof} 

In the CPV, there is no way to change $p_{SID}$ from the operator side. Hence, we cannot eliminate the NE that results in a cheating outcome. However, in the PPV, since we assume $p_{SID}$ increases with the increase in $VVMT_{thld}$, we can change $VVMT_{thld}$ to increase the proportion of SIDs in the eligible voters and eliminate the cheating-outcome NE, which proves the advantage of PPV. 

However, in PPV, it is also essential to balance to number of eligible voters and $VVMT_{thld}$. If $VVMT_{thld}$ is too high, there will be only a small number of eligible voters, then Assumption \ref{assump_int_cost} may not hold given large variance in a small population (i.e., it is possible that some attackers are included into this small group of eligible voters).

\section{Experiments and evaluation}\label{sec:exp}
To evaluate the performance of the proposed POT protocol and the extended V2I voting algorithms, we implement the POT protocol and the voting algorithm in the V2X Simulation Runtime Infrastructure (VSimRTI)~\cite{protzmann2017simulation} and conduct vehicular simulations in a two-lane highway (Interstate highway I-80 within Wyoming).

The simulation is conducted in a desktop with Intel(R) Core(TM) i7-6850K processor at 3.6GHz, 32GB RAM memory, and (X64)Windows 10 operating system. The parameter settings for the simulation are given in Table~\ref{tab3}.

We are interested in the security and performance of the POT-enhanced voting algorithms (PPV Algorithm 2) by using Plurality-voting (CPV Algorithm 1) as the benchmark. For the security metric, the proportion of invalid events is used to evaluate to what extent each type of voting algorithm is vulnerable to malicious reports. For the performance metric, the throughput by the RSU running voting algorithms with respect to the number of correct events per minute is used for the evaluation.

Additionally, we also consider the time it takes for the RSU to run the voting algorithms to confirm an event in each round (i.e., latency), which can be influenced by the minimum number of votes required and the local traffic density, as suggested by the simulation results.

When evaluating the PPV algorithm, we take into account the influence of the proportion of SIDs ($p_{SID}$) on algorithm security and performance by varying the threshold of voting $VVMT_{thld}$. Therefore, all simulation results include scenarios when CPV is adopted and PPV with $p_{SID}$ being equal to $30\%$, $50\%$, and $70\%$. 

All the experiments are conducted under different traffic settings where the traffic density is in the range between 35 and 45 vehicles per mile, which corresponds to the normal flow conditions in which traffic demands are approaching or equal to roadway capacities in U.S. roadways~\cite{manual2010hcm2010}. 

\begin{table}[t]
	\caption{PARAMETERS SETTINGS IN THE SIMULATION}
		\begin{tabularx}{\columnwidth}{c|X}
			\hline
			\textbf{Parameters}&\textbf{Settings} \\
			\hline
			\hline
			Simulation time& 1000 seconds \\ 
			\hline
			Max number of vehicles& 2000\\ 
			\hline
			Speed of vehicles& 60 mph (max)\\ 
			\hline
			Traffic flow& 600,1200,2500,3600\\ 
			\hline
			Traffic simulator& SUMO 1.2.0 \\ 
			\hline
			Network simulator& Simple Network Simulator (SNS) \\
			\hline
			Communication range of OBU/RSU& 500 meters\\ 
			\hline
			Hash algorithm used by RSU& SHA-256 \\
			\hline
		 \end{tabularx}
		\label{tab3}
\end{table}

\begin{figure*}[tb!]
\centering
\begin{subfigure}{0.46\textwidth}
  \centering
  \includegraphics[width=\textwidth]{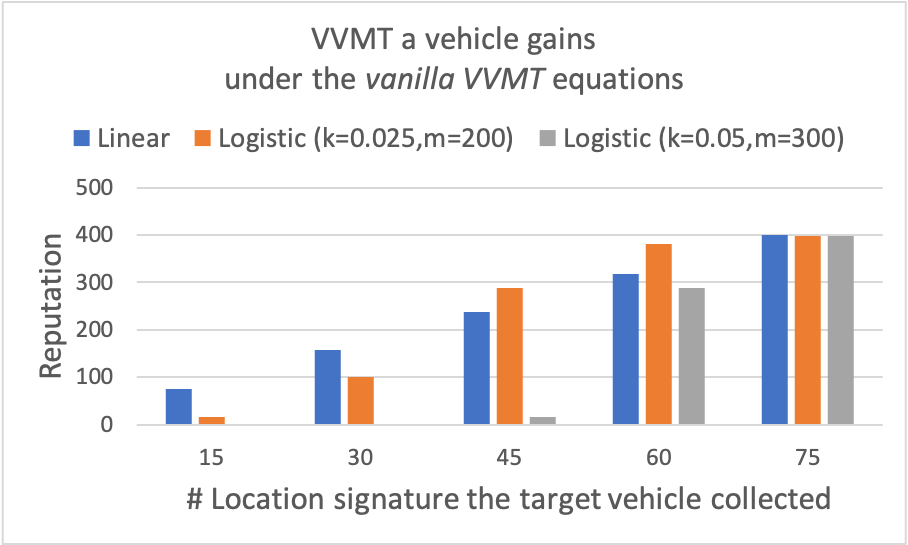}
  \caption{\textit{Vanilla VVMT}}
  \label{fig:vanilla_vvmt}
\end{subfigure}
\begin{subfigure}{0.46\textwidth}
  \centering
  \includegraphics[width=\textwidth]{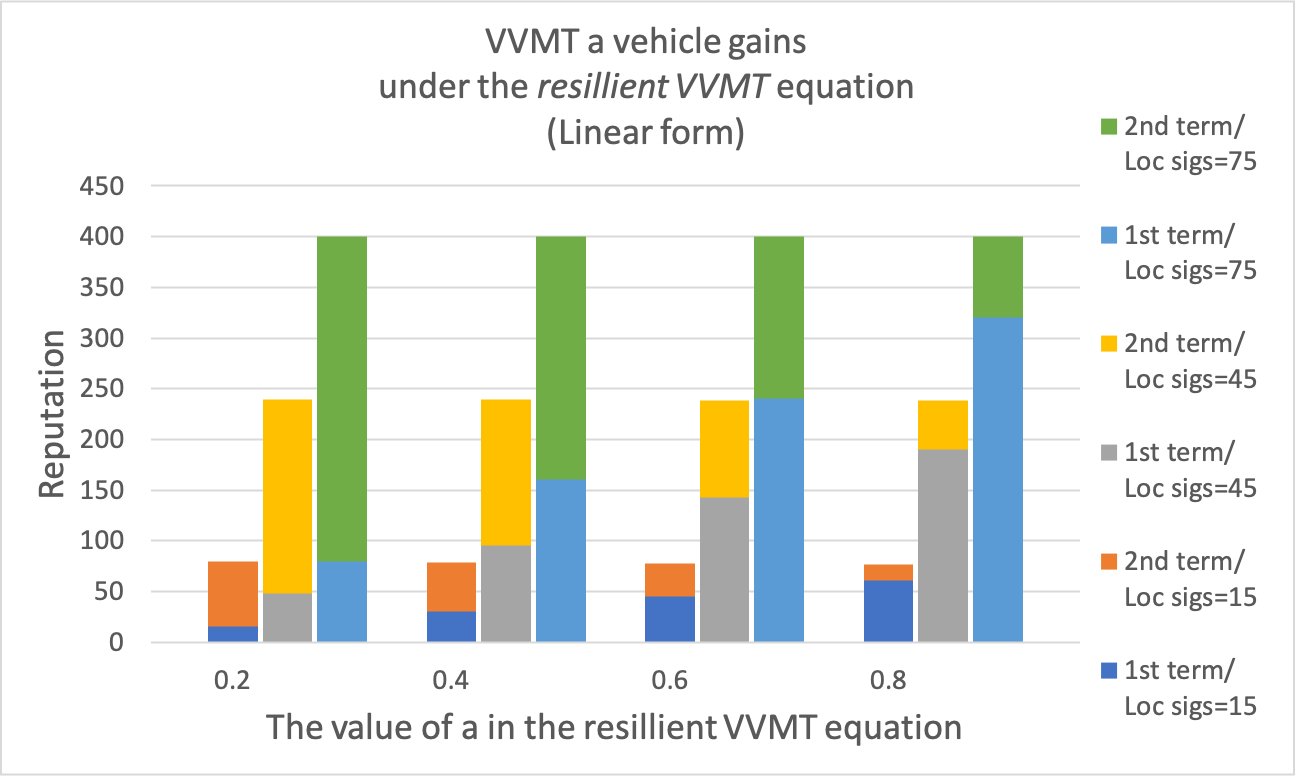}
  \caption{\textit{Resilient VMT} in linear form}
  \label{fig:resillient_vvmt_linear}
\end{subfigure}
\begin{subfigure}{0.46\textwidth}
  \centering
  \includegraphics[width=\textwidth]{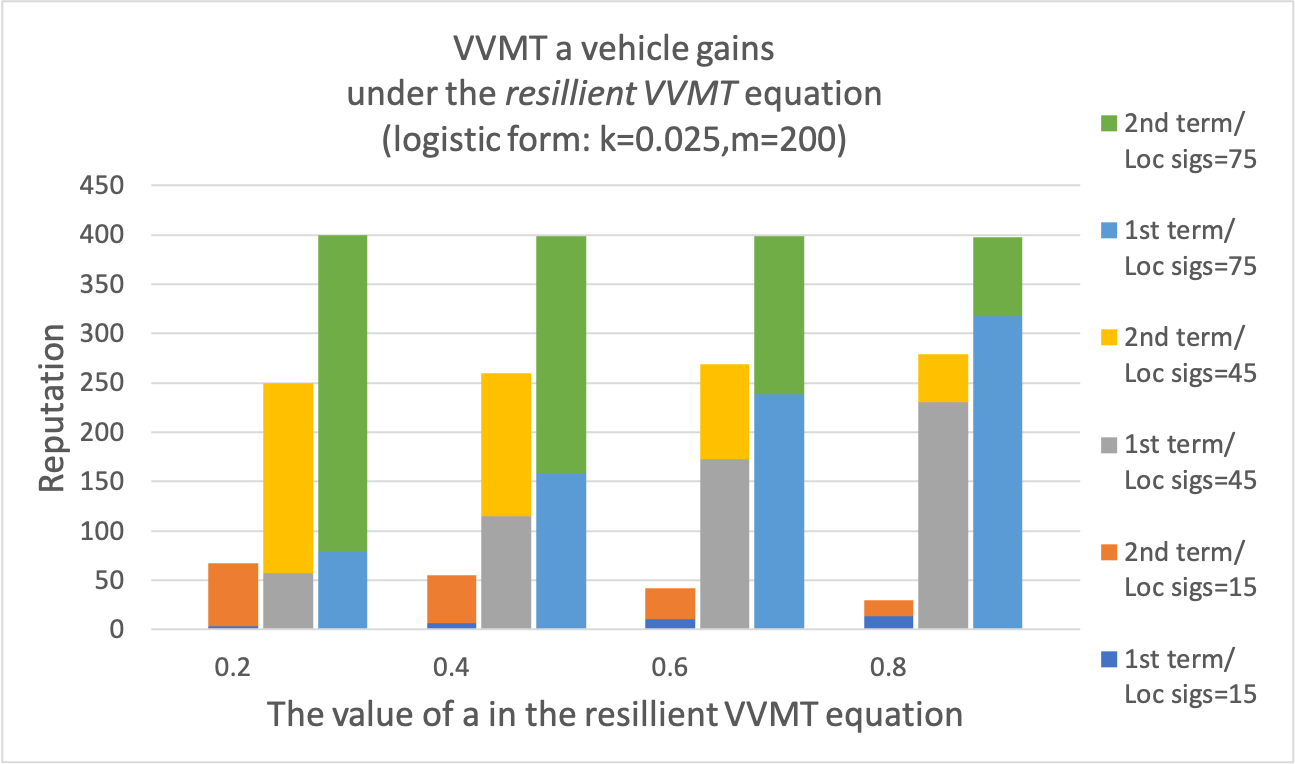}
  \caption{\textit{Resilient VVMT} in logistic form:k=0.025,m=200}
  \label{fig:resillient_vvmt_logk025m200}
\end{subfigure}
\begin{subfigure}{0.46\textwidth}
  \centering
  \includegraphics[width=\textwidth]{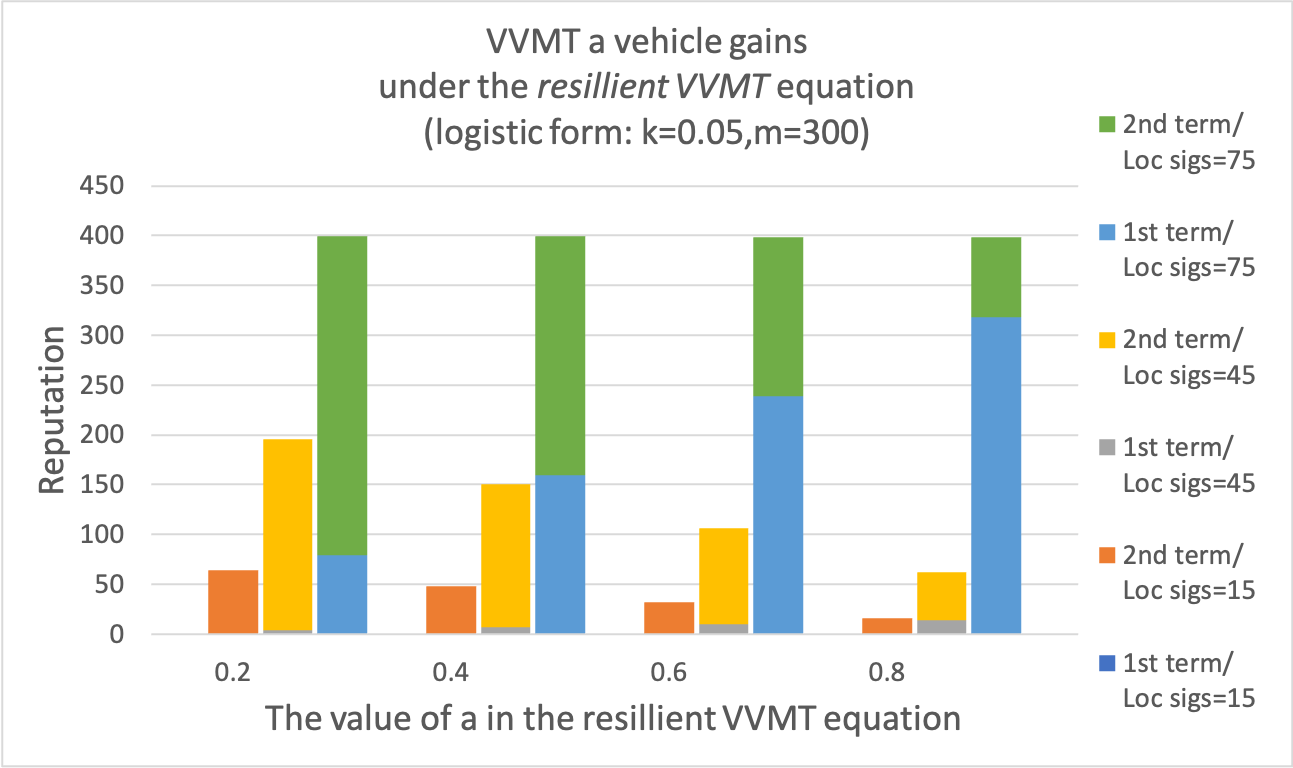}
  \caption{\textit{Resilient VVMT} in logistic form:k=0.05,m=300}
  \label{fig:resillient_vvmt_logk05m300}
\end{subfigure}
\caption{Vehicle reputation under different VVMT equations.}
\label{fig:vvmt_reputation}
\end{figure*}

\subsection{The influence of VVMT designs on vehicle reputation}
A key design aspect of the proposed POT protocol is the influence of different implementations of VVMT, namely eq.~\ref{eq:vanilla_VVMT} and \ref{eq:flexible_vvmt} on vehicle reputation. We present a case study on the proof-of-travel collection in a real-world interstate corridor in Southern Wyoming, USA. According to the U.S. Department of Transportation (DOT), interstate 80 (I-80) was deployed with 75 RSUs along a 400-mile road segment~\cite{pilotconnected}. Connected vehicles that are installed with onboard V2X units can share with RSUs information regarding road and weather conditions captured by on-vehicle sensors.

When the \textit{vanilla VVMT} (eq.~\ref{eq:vanilla_VVMT}) is adopted, a vehicle gains a reputation that corresponds to the physical distance, as shown in Fig.~\ref{fig:vanilla_vvmt}. Rather than having the vehicle accumulate its reputation with linear growth as it travels (orange bar in Fig.~\ref{fig:vanilla_vvmt}), we may modify parameters ($k$ and $m$) in the logistic function form in eq.~\ref{eq:vanilla_VVMTb} to postpone the time the vehicle starts to gain reputation or the difficulty, as shown in the orange and grey bars in Fig.~\ref{fig:vanilla_vvmt}. In this case, drivers who join the V2X ecosystems later may find it more difficult to gain reputation values as rewards, which encourages the early adoption of V2X technologies. As mentioned earlier, this scenario is inspired by the incentive designs in existing Blockchain-based cryptocurrencies. For example, similar to the dynamic mining difficulty, a mechanism that constantly changes parameters in Bitcoin or Ethereum to control the average time to mine a block, we can use different functions or parameters in the \textit{vanilla VVMT} equation to control how difficult or fast each vehicle can gain reputation, as shown in different curves for reputation growth in Fig.~\ref{fig:vanilla_vvmt}.

While the \textit{vanilla VVMT} requires a vehicle to collect every location signature from all the RSUs it meets when traveling, the \textit{resilient VVMT} (eq. 3) provides fault tolerance by allowing the vehicle to miss some RSUs along the path of its movement. When eq.~\ref{eq:flexible_vvmt} is adopted, both the distance (the 1st term) and the ratio of the actual number of location signatures to the total number of RSUs (the 2nd term) will contribute to the final scores of vehicle reputation, as shown in Fig.~\ref{fig:resillient_vvmt_linear},~\ref{fig:resillient_vvmt_logk025m200}, and~\ref{fig:resillient_vvmt_logk05m300}. The horizontal axis in these three figures represents the weights (i.e., $\alpha$ in eq.~\ref{eq:flexible_vvmt}) of the first and the second term.

The trend of reputation growth shown in Fig. 5b is similar to Fig. 5a (blue columns) in which vehicle reputation scores increase steadily regardless of the contributions made by the first and the second terms in eq. 3a. The result implies a lack of flexibility in using the linear versions to implement both the vanilla VVMT (eq. 2a) and the resilient VVMT (eq. 3a) equations.

On the other hand, the logistic versions of both the vanilla and resilient VVMT equations are much more flexible and can support the real-time adjustment of the difficulty level for vehicles to gain a reputation. For the logistic versions (eq. 2b and eq. 3b), we can increase the values of parameters k and m to make it more difficult to increase VVMT. For example, when k=0.025 and m=200, a vehicle that has collected 45 locations signatures along its path of movement (assuming there are 75 roadside units in total) may earn 250 reputation scores or more, as shown in Fig. 5c (the middle columns with yellow and grey colors). On the other hand, when k=0.05 and m=300, the reputation score earned by a vehicle with the same number of location signatures (i.e., 45) has only an average value of 100, while reaching the minimum score of 50, as shown in Fig. 5d.

The results presented earlier are aligned with our design rationales regarding the role of k and m in the proposed Proof-of-Travel protocol. Similar to the existing Blockchain-based cryptocurrencies in which a system-defined target value (e.g., 256-bit number) governs the difficulty of successful “mining”, in the Proof of Travel protocol, we would also like to have an adaptive mechanism to adjust how difficult vehicles earn reputation at different stages of V2X deployment.


\subsection{The security of voting algorithms}
The security metric of the proportion of invalid events is derived by comparing the number of times when the RSU confirms false events or fails to report a true event that actually occurred (because no consensus is formed among vehicles based on a given voting rule) against the number of times when the RSU confirms and disseminates correct events. We refer to the former two types of events as "invalid events" because each of them represents an undesired situation we are trying to avoid: while confirming and disseminating a false traffic event means that the traffic management center will mistakenly allocate valuable resources, such as first responders or emergency medical services, failing to confirm a true event can result in the loss of assets or life in the event of man-made or natural disaster. 

The proportion of invalid events derived can then be used to evaluate the extent to which CPV or PPV is vulnerable to malicious reports, as shown in Fig~\ref{fig:voting_security}. The horizontal axis represents the proportion of the malicious when the vehicular network starts to operate. We compare the security performance between the CPV and PPV by evaluating the proportion that the RSU failed to confirm correct events. Also, we adjust the number of votes and voting threshold in PPV to get a better understanding of PPV's performance under different parameter settings. The results suggest that PPV is less vulnerable to malicious votes than CPV even with the existence of a higher percentage of malicious vehicles when the network starts to operate. The proportion of invalid events by CPV (denoted by blue shaded curves in Fig~\ref{fig:voting_security}) increases to more than 0.8 as the percentage of malicious vehicles increases from 10$\%$ to 50$\%$, as shown in Fig~\ref{fig:voting_security_min5},~\ref{fig:voting_security_min7},~\ref{fig:voting_security_min10}, and~\ref{fig:voting_security_min12}. However, as long as we can adjust the voting threshold $VVMT_{thld}$ in the PPV algorithm for increasing the proportion of SIDs to between (based on Assumption VI.1 and Proposition 2), we can reduce the chance of RSU's failures in confirming correct events to the range between $30\%$ and $60\%$ with a small number of votes required $N_{thld}$.

\begin{figure*}[tb!]
\centering
\begin{subfigure}{0.48\textwidth}
  \centering
  \includegraphics[width=\textwidth]{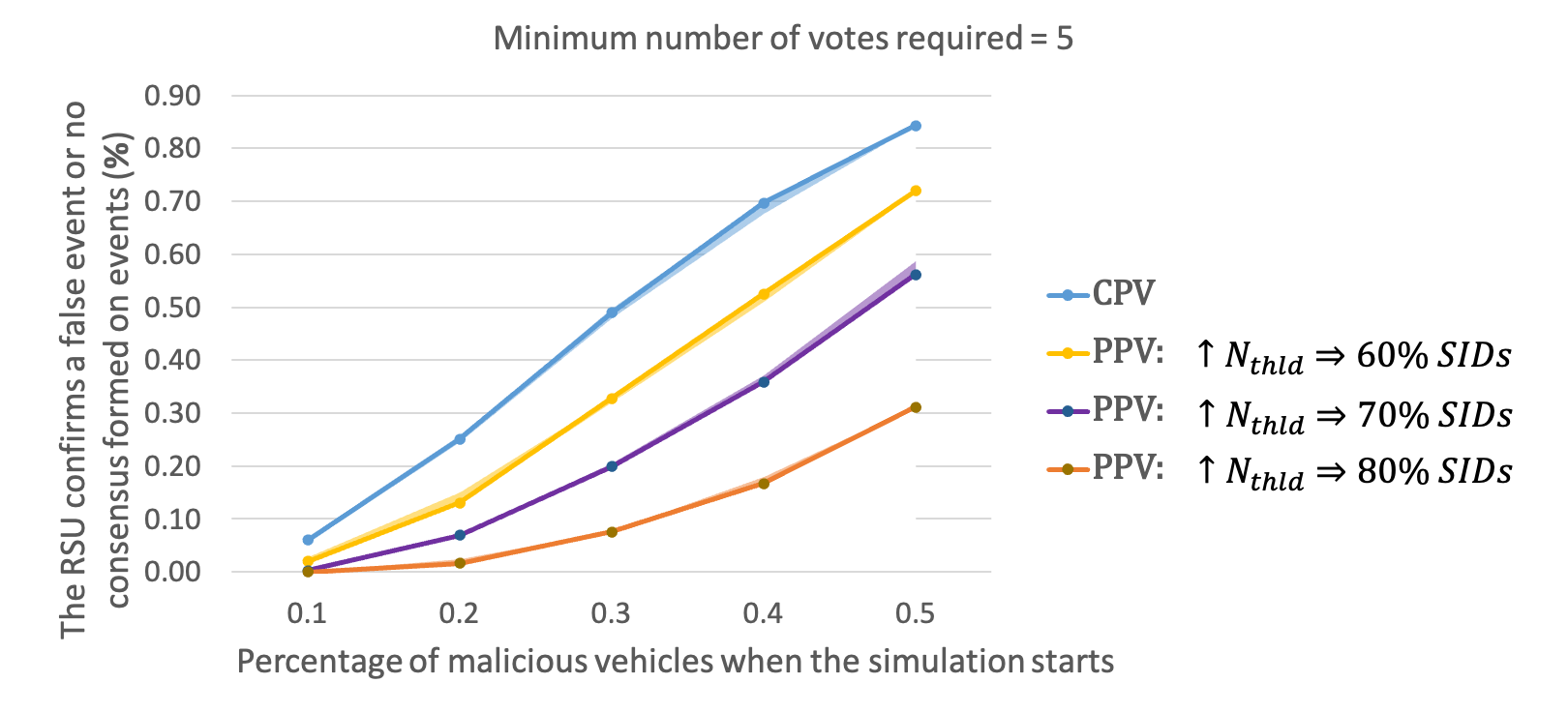}
  \caption{The required minimum number of votes $N_{thld} = 5$}
  \label{fig:voting_security_min5}
\end{subfigure}
\begin{subfigure}{0.48\textwidth}
  \centering
  \includegraphics[width=\textwidth]{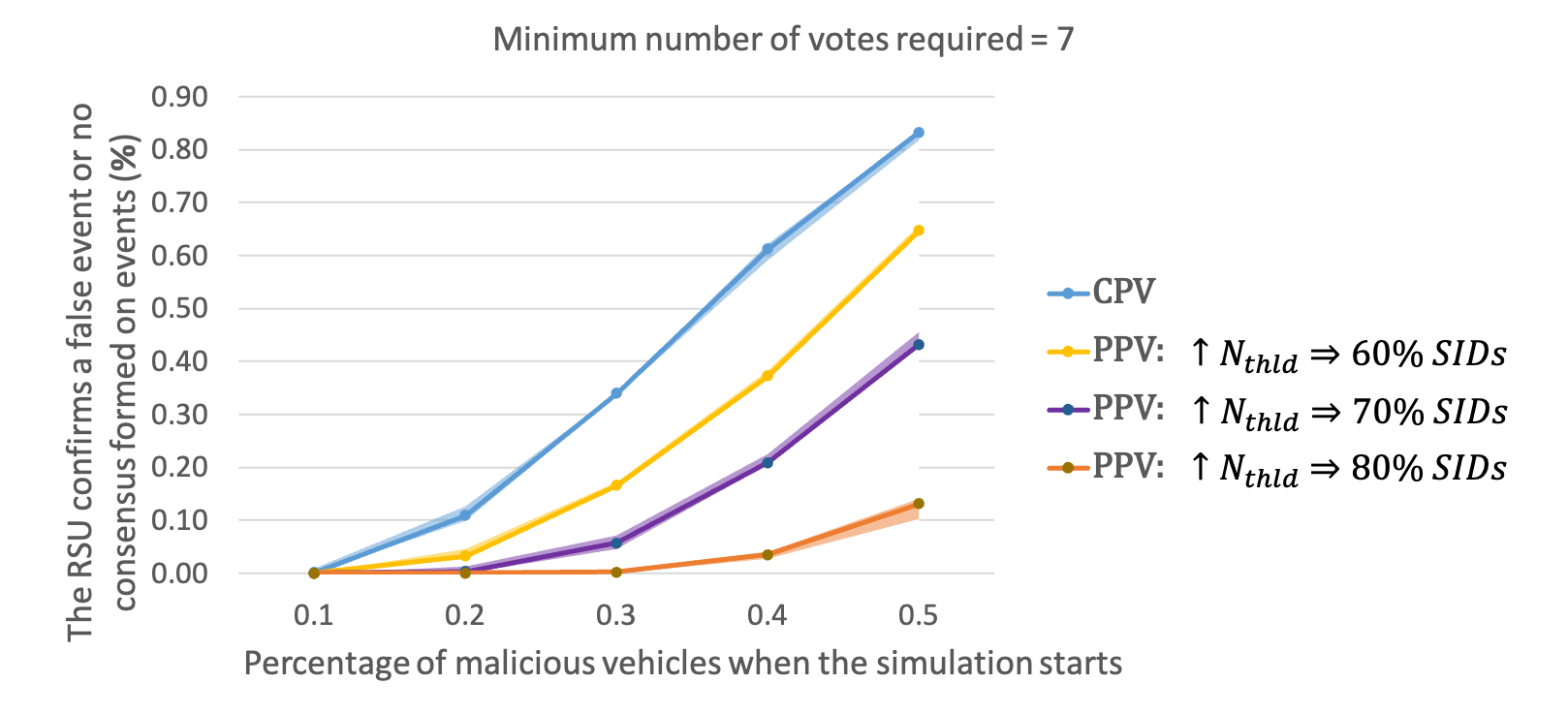}
  \caption{The required minimum number of votes $N_{thld} = 7$}
  \label{fig:voting_security_min7}
\end{subfigure}
\begin{subfigure}{0.48\textwidth}
  \centering
  \includegraphics[width=\textwidth]{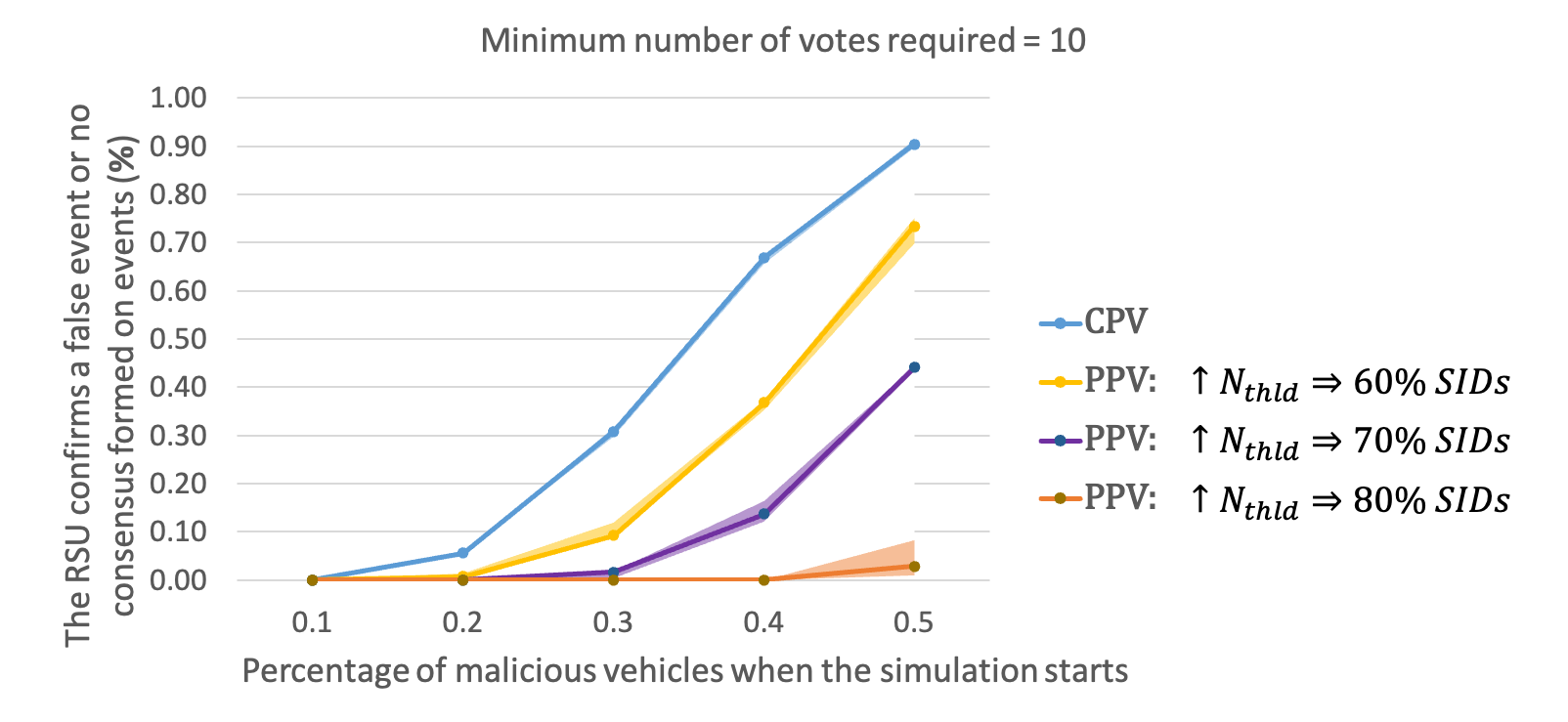}
  \caption{The required minimum number of votes $N_{thld} = 10$}
  \label{fig:voting_security_min10}
\end{subfigure}
\begin{subfigure}{0.48\textwidth}
  \centering
  \includegraphics[width=\textwidth]{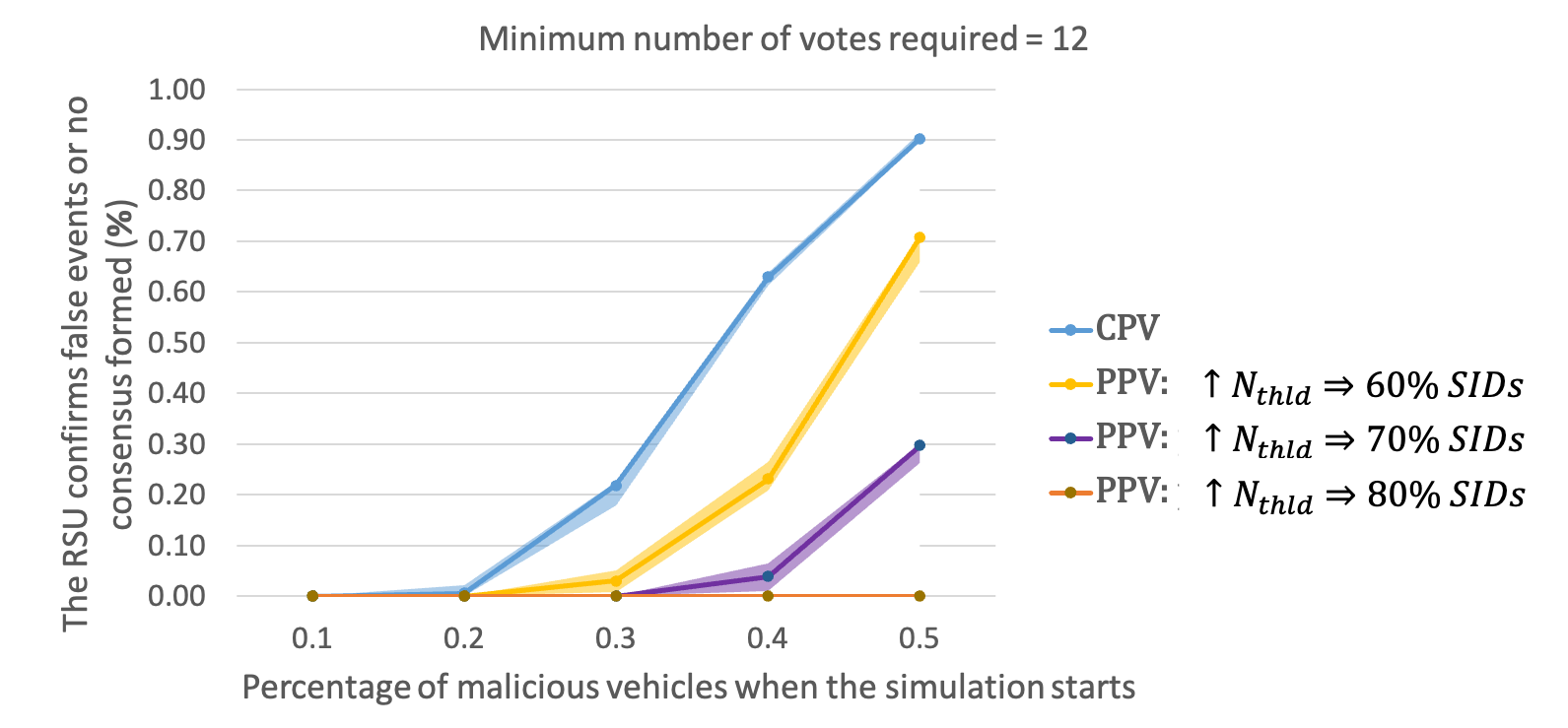}
  \caption{The required minimum number of votes $N_{thld} = 12$}
  \label{fig:voting_security_min12}
\end{subfigure}
\caption{Security vulnerabilities to malicious vehicles: CPV vs. PPV.}
\label{fig:voting_security}
\end{figure*}

\subsection{The tradeoff between event criticality and voting delay}
For certain V2I-enabled crowdsensing applications that are safety-critical, it is desirable for the proposed voting algorithms to achieve a higher percentage of confirming correct events. This can be realized by increasing the number of votes required $N_{thld}$ or the voting threshold $VVMT_{thld}$ in PPV. However, these efforts will result in a delay in event confirmation time based on our experiments. Therefore, it is in the interest of security engineers to evaluate this tradeoff situation between security with respect to ensuring correct decisions on events and timing regarding making decisions on whether to respond to a vehicle-reported traffic event in time.

We present a tradeoff graph for choosing parameter values to balance the two conflicting needs, as shown in Fig.~\ref{fig:voting_latency}. When the traffic density is small (e.g., $<$20 vehicles per mile), it is unwise to set a large number of required votes $N_{thld}$ (e.g. $>$10), which will result in an extremely high delay (red color cells in Fig.~\ref{fig:voting_latency}) in confirming an event. For example, setting $N_{thld}$ to 12 in PPV can lead to a confirmation delay of more than one minute. While the delay time is tolerable for the weather and road services discussed early, it might be unacceptable when V2I-enabled crowdsensing is used during the emergency evacuation. The same design rationale can be applied to the local sensing and dynamic map services. For building HD maps, we may set a high value for $N_{thld}$ to ensure the correctness of unsynchronized sensing data shared by different vehicles.

\begin{figure}[tb!]
\centering
    \includegraphics[width=0.45\textwidth]{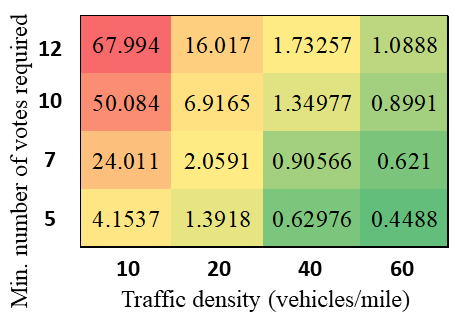}
\caption{Event confirmation time (seconds).}
\label{fig:voting_latency}
\end{figure}

The same tradeoff exists between the parameter of the voting threshold $VVMT_{thld}$ and the security of the voting game. In most of the simulation runs of our experiments, we observe around 20$\%$ reduction in event-confirmation time by lowering the voting threshold $VVMT_{thld}$ to make 10$\%$ extra vehicles eligible to vote. However, this also means that the PPV algorithm will have less chance of confirming correct events, as the results shown in the previous section. We recommend that each tradeoff decision in choosing parameters in PPV be driven by the criticality (i.e., how much tolerance we have regarding event correctness) and timing constraints (how quickly the event needs to be confirmed and routed) for the specific crowdsensing application. 


\section{Design considerations for the real-world deployment of the POT protocol}
\subsection{Security credential management}
To adopt and implement the POT protocol in the real world, we need to consider the design opportunities posed by credential management systems for V2X communication and the constraints by the distribution of infrastructure components (e.g., RSUs) in the target region and country for V2X deployment.  

The security credential management system (SCMS) developed for connected vehicles in the U.S.~\cite{CVCredential,brecht2018security}, for example, provides opportunities for POT applications. Since one of the goal of the SCMS is to identify and revoke certificates assigned to compromised vehicles that broadcast malicious events, we can explore the role of the POT protocol in this revocation process. Specifically, it is proposed in~\cite{CVCredential} that the SCMS uses malicious reports sent by onboard devices installed on vehicles to determine whether to revoke the certificate assigned to a particular vehicle. In this case, a vehicle's reputation attested by its proof of travel can serve as the root of trust in its malicious reports in the early stage of SCMS deployment. Additionally, the idea of providing reputational rewards to vehicles that share traffic event data can be extended and applied to the SCMS for incentivizing each connected vehicle to share malicious behaviors they detected to infrastructure. 

Once the reputational points a vehicle owns are recognized and stored on SCMS infrastructure, the records must be protected from malicious tampering or unauthorized access. In this regard, blockchain technology is an ideal candidate for protecting these chains of “travel proof” or reputation scores, as discussed in~\cite{liu2018blockchain,kang2019toward,yang2018blockchain,liu2019blockchain}.

Among design constraints for implementing the proposed POT protocol in the real world, the most challenging one faced by engineers might be to decide the maximum length of the chain of proofs a vehicle should store on board. For the convenience of identity verification, a vehicle might be required by infrastructure to keep a minimum number of location signatures on board and attach them to V2I messages for identity verification. For example, proofs of locations have been used to detect Sybil nodes in vehicular networks~\cite{park2013defense}. Previous work shows that each vehicle must be required to acquire at least 15 location signatures from RSUs along its path of movement if we want to achieve 98$\%$ accuracy in detecting Sybil nodes~\cite{chen2009robust}. However, this security guarantee comes at a price: maintaining a long chain of location signatures will increase the memory consumption by each vehicle and incur extra communication overhead. 

\subsection{Privacy}
While deploying the POT protocol, the privacy of vehicle owners, especially location privacy, needs to be protected during the proof-collection process. In particular, previous work uses threshold ring signatures to achieve conditional privacy in authenticating V2I messages~\cite{liu2019blockchain,li2018creditcoin}. Another promising technique is zero-knowledge proof~\cite{wang2020hdma}, which allows mutual authentication between vehicles without revealing their identities.

\subsection{The influence by the criticality of V2I data}
Although this paper focuses on rewarding vehicles with reputational scores based on travel proof, the protocol allowing each vehicle to gain reputation through V2I data sharing can be extended to consider data with different criticality levels. The idea is to reward a vehicle with extra points in addition to reputational points derived from travel proof if the sensing data the vehicle shared has high criticality levels (e.g., high-criticality such as incident reports vs. medium-criticality weather conditions). Specifically, the critical level of a given type of V2I data can be influenced by technological, policy, and business factors. 
\begin{itemize}
    \item First, from the perspective of the data receiver, the timing and the location of the vehicle-reported event may decide how critical it is. For example, for a vehicle receiver, a traffic event a few miles away will have less influence on its maneuvers than an event that occurs within the same road segment. Therefore, V2I data indicates an event closer to the receiver might be regarded by the receiver as more "critical."
    \item Second, a given standardization body might assign higher priority to certain types of events. For example, when it comes to dedicated-short range communication (DSRC), the U.S. Federal Communication Commission (FCC) has previously reserved dedicated channels for V2X applications involving incident mitigation and public safety while assigning a service channel for the data exchange about convenience and long-distance V2X applications~\cite{zhang2012vehicle}. Although this paper does not make assumptions about the physical links to realize V2X communication (e.g., DSRC vs. Cellular-V2X), security engineers who want to implement the POT protocol in intelligent transportation applications should be aware of the potential influence of standardization efforts in their region.  
    \item Third, in addition to technological and policy factors, the criticality of a given type of V2X data can be affected by how each stakeholder in the V2X ecosystem valuates the use cases enabled by the data type. For example, while emergency responders or individual drivers might prioritize the V2I data about incident location and severity, a commercial freight company might be more interested in V2I information about route guidance, road, or even weather conditions.
\end{itemize}

\subsection{Fairness in the voting game: plurality vs. weighted voting}
The proposed PPV algorithm considers all vehicles that are eligible to vote to have equal weights in voting. The main purpose is to reduce the influence of "high-reputation" vehicles that are compromised by adversaries on voting results. For example, taxis and trucks can gain higher VVMT values due to long hours of driving per day. Also, special vehicles, such as police cars, fire trucks, or ambulances might be assigned with high reputation than private passenger vehicles.

Alternatively, we may adopt weighted voting to achieve fairness in voting by giving high-reputation vehicles a higher stake in the voting. Although studies show that weighted voting has the potential to improve algorithm efficiency and fairness, it is more vulnerable to adversarial manipulation than plurality (i.e., equal weight) voting when vehicles with high VVMT are compromised by adversaries~\cite{aziz2011false,leonardos2020weighted}. This becomes possible when special vehicles, such as police cars with V2X connectivity, become the target of attacks~\cite{policecarhack}. Our future work will explore how to achieve the tradeoff between fairness and security for weighted voting algorithms.

\subsection{A smooth transition from low to high V2X penetration}
While the security solutions presented in this paper target more on the transitional period, the trust management mechanisms also need to be adapted to evolving usage scenarios as V2X technologies mature. We suggest that security engineers take into account the following aspects when designing trust-based security schemes for future V2X markets with high maturity levels. 

Regarding the maturity level of V2X, we envision that it can be characterized by the level of penetration rate (or the number of users), the number and types of applications that V2X supports, and the geographical regions covered by V2X communication. Therefore, we suggest that trust management and data validation also evolve along with these three factors. 

First, regarding the V2X penetration rate and the number of users, we envision that a large number of vehicles built by different original equipment manufacturers (OEMs) will be pre-installed with V2X modules. This differs from current situations where a lot of traditional vehicle models are integrated with aftermarket V2X devices when individual vehicle owners decide to participate in V2X pilot projects led by public transportation agencies or a company opt-in testing V2X technologies. Therefore, a lot of improvements can be made in the future in the areas of managing and distributing vehicle V2X credentials. This calls for more transparency in the supply chain of connected vehicles, especially during the registration, distribution, and “recycling” of V2X credentials. All of these solutions require that OEMs, wireless communication service providers, and public transportation agencies work closely and establish a connected-vehicle security-sharing sharing platform in which information about vehicles misbehaviors, the status of vehicle credentials, and even the trust or reputation scores proposed and discussed in this paper can be accessed in a secure and privacy-preserved manner.

Second, as the number and types of V2X applications increase, there will be various types of V2X data modalities to which we can get access. For example, With V2X infrastructure, which provides large bandwidth and low latency, multiple vehicles can share high-dimensional raw sensing data (e.g., point cloud) collectively in real-time. Although this paper focuses on single-modality location-based data, which contain basic information about vehicles and traffic, there can be multimodal data, such as 2D images, point cloud, or audio, shared through V2X channels. While providing more detailed information about traffic participants and environmental conditions, such new data modalities also call for new methods for trust establishment and data validation.

Third, while solutions discussed in this paper can work at the regional or state level, we envision there will be a need for trust authorities that can work across the boundaries between states, regions, or even countries. Although researchers have considered the use of hierarchical trust authorities for managing certificates and trust information for vehicles operating in different regions, it can be challenging to deal with regions with limited wireless coverage or even rural areas. More work needs to be done regarding ensuring trust and preventing V2X malicious behaviors in these corner cases. 

\section{Conclusion}
This paper proposes a V2I communication protocol, titled Proof of Travel, to determine each connected vehicle's reputation by verifying its spatial movement with cryptography methods. Based on the gained proof of each vehicle, we proposed a PPV algorithm to enhance previous plurality voting algorithms (e.g., CPV) for validating V2I-reported traffic events. By using a game-theoretic framework, we also prove that, when PPV is adopted, all vehicles choosing to cheat is not a pure Nash equilibrium. The reason is that rational adversaries governed by profit-seeking behaviors will lose interest in gaining proof of travel for launching attacks if the extra cost incurred by proof of travel exceeds the adversarial reward. Results from simulation experiments also suggest that the PPV algorithm can tolerate a higher proportion of malicious vehicles and generate higher throughput than the benchmark CPV algorithm. Here, we focus on the adversaries who are "rational" and lose interest in gaining proof for launching attacks. Future work will tackle the type of adversaries with "irrational" behavior and unlimited resources.

\appendix
\begin{prop}\label{prop_POT_secure}
Under the assumption that drivers are rational in an economic sense, i.e., those whose rewards in voting are bounded by a value are trying to maximize their payoffs, increasing $VVMT_{thld}$ can result in higher proportion of SIDs in the eligible voter group. In other words, a vehicle $i$ with higher $VVMT_i$ is more likely to be owned by a normal rather than a malicious driver.
\end{prop}

\begin{figure*}[tb!]
\centering
\begin{subfigure}{0.32\textwidth}
  \centering
  \includegraphics[width=\textwidth]{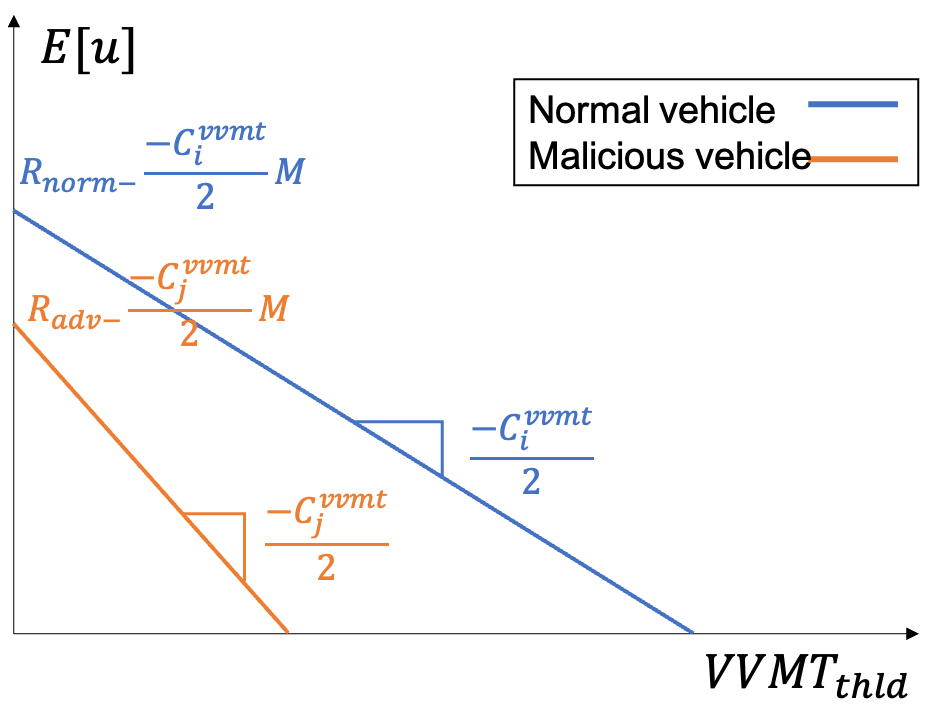}
  \caption{$R_{norm} \geq R_{adv}$}
  \label{fig:exp_payoff_a}
\end{subfigure}
\begin{subfigure}{0.32\textwidth}
  \centering
  \includegraphics[width=\textwidth]{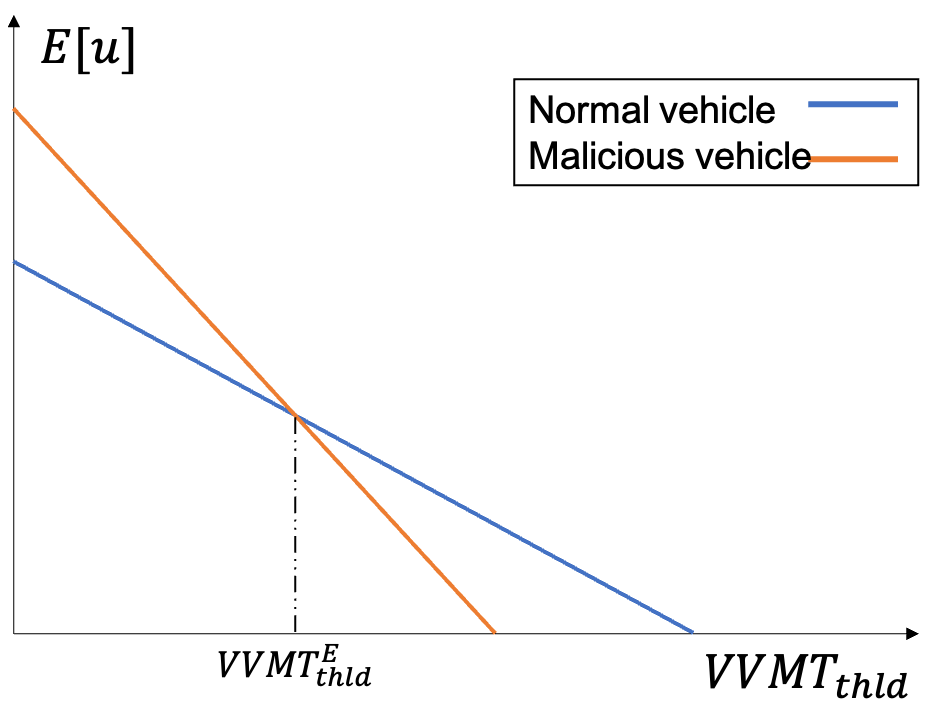}
  \caption{$R_{norm} < R_{adv}$}
  \label{fig:exp_payoff_b}
\end{subfigure}
\begin{subfigure}{0.32\textwidth}
  \centering
  \includegraphics[width=\textwidth]{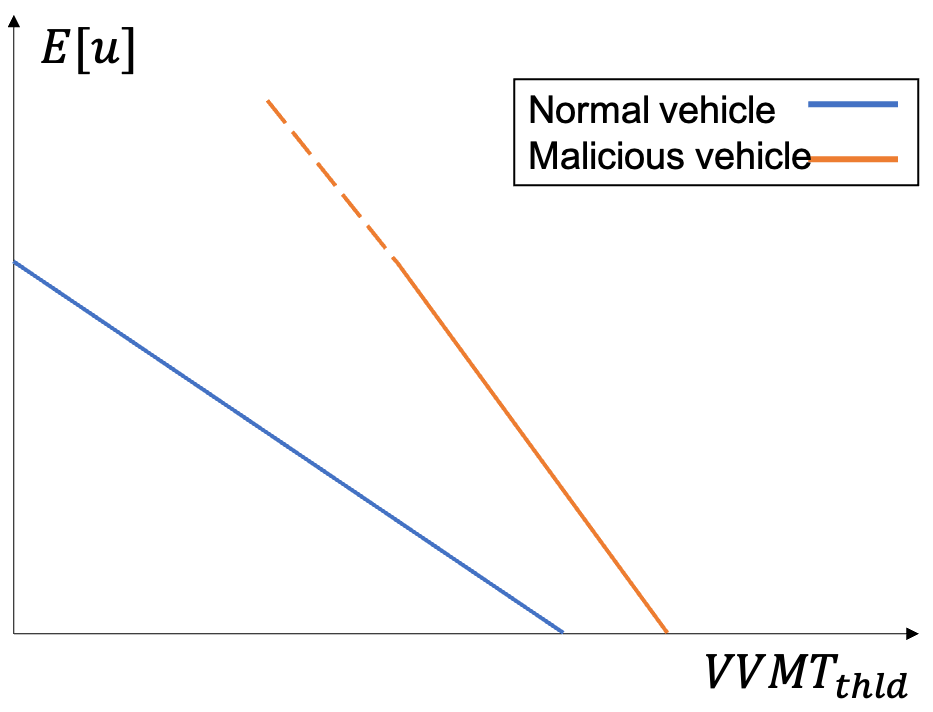}
  \caption{$R_{norm} << R_{adv}$}
  \label{fig:exp_payoff_c}
\end{subfigure}

\caption{The changes in the expected payoffs of normal and malicious vehicles.}
\label{fig:exp_payoff}
\end{figure*}

\begin{proof}
Consider the scenario where vehicle $i$ is owned by a normal driver and vehicle $j$ owned by a malicious driver, we then have the cost of travel per VVMT by $vehicle_i$ less than the cost by vehicle $j$, i.e., $C^{vvmt}_i < C^{vvmt}_j$. The reason is that a normal driver does not have to take the same amount of burden a malicious driver takes to acquire VVMT. In addition to energy consumption for traveling, the malicious driver, whose only goal is to compromise the vehicular network, has to take extra efforts to get access to valid V2X credentials stored in compromised vehicles, while a vehicle driven by the normal driver will naturally move from its origin to destination. 

For vehicle $i$ and $j$, each one must have a positive payoff for having the incentive to participate in voting, which is given in eq.~\ref{payoff_normal} and~\ref{payoff_malicious}.

\begin{equation}\label{payoff_normal}
	U_i = R_{norm} - C^{vvmt}_{i}*VVMT_i 
\end{equation}

\begin{equation}\label{payoff_malicious}
	U_j = R_{adv} - C^{vvmt}_{j}*VVMT_j 
\end{equation}

It should be addressed that the notations that we used here to represent vehicle payoffs ($U_i, U_j$),reward ($R_{norm},R_{adv}$), and cost per verifiable miles of travel ($C_i^{vvmt},C_j^{vvmt}$) are different from those used in the security analysis of the voting games (Section VI-B) in the sense that they refer to different phases of travel, as shown in Fig.~\ref{fig2}. The security analysis presented in section VI-B mainly focuses on actions by each player who has already made efforts to gain enough \text{VVMT} and is thus eligible to vote (right of Fig.~\ref{fig2}). Here, we are interested in understanding how increasing or decreasing $VVMT_{thld}$ can influence the decision-making process by each vehicle before the voting actually occurs (left of Fig.~\ref{fig2}): is it worth traveling certain miles while following the POT protocol to gain reputation if "my" only goal is to send fake V2I messages to compromise the network services? The hope is that by increasing the threshold $VVMT_{thld}$ of voting eligibility, fewer malicious vehicles which originally just wanted to misbehave will continue while we end up with more SID vehicles left.  

For vehicle i, without loss of generality, we assume that its accumulated VVMT is drawn from a uniform distribution $VVMT_{i} \sim U(0,M)$, where $M$ denotes the maximum VVMT a vehicle may acquire. Since only the vehicles that have $VVMT_i > VVMT_{thld}$ are eligible to participate in voting and thus earn rewards, the expected payoff for vehicle $i$ can be denoted by eq.~\ref{payoff_exp_norm}.

\begin{align}\label{payoff_exp_norm}
	E[U_{i}] & = \int_{VVMT_{thld}}^{M} (R_{norm}-C^{vvmt}_{i}*VVMT_i) \nonumber \\
	           & \quad \, \frac{1}{M-VVTM_{thld}}d(VVMT_i) \nonumber \\
	           & = R_{norm} \nonumber \\ 
	           & \quad \, -  \frac{C^{vvmt}_i}{2(M-VVMT_{thld})} (M^2-VVMT^2_{thld}) \nonumber \\
	           & =  R_{norm} - \frac{C^{vvmt}_i}{2}(M+VVMT_{thld})
\end{align}

Following the same procedure, the expected payoff for malicious vehicle $j$ can be derived and is given in eq.~\ref{payoff_exp_adv}.

\begin{align}\label{payoff_exp_adv}
	E[U_{j}] & = R_{adv} - \frac{C^{vvmt}_j}{2}(M+VVMT_{thld})
\end{align}

Based on the values of $R_{norm}$ and $R_{adv}$, the impact of increasing the voting threshold $VVMT_{thld}$ on the likelihood that each type of vehicle will choose to opt in or out of the voting game can be classified into three scenarios, as shown in Fig.~\ref{fig:exp_payoff}.

\textbf{$R_{norm} \geq R_{adv}$}: a normal vehicle always have a stronger incentive to opt-in than its malicious counterpart regardless of the value of $VVMT_{thld}$, as shown in Fig.~\ref{fig:exp_payoff_a}.  

\textbf{$R_{norm} < R_{adv}$}: a normal vehicle will have a stronger incentive to opt-in than its malicious counterpart as long as the voting threshold $VVMT_{thold}$ is greater than a certain value, denoted as $VVMT_{thld}^E$, as shown in eq.~\ref{vvmt_thld_critical}. Although the payoff by the malicious vehicle is higher than the normal one when $VVMT_{thld}$ is small, the former drops at a faster rate as $VVMT_{thld}$ increases. This is understandable as the cost per VVMT for the malicious party is higher. $VVMT_{thld}^E$ can be derived by letting $E[U_i]=E[U_j]$, the critical point when their payoffs become the same. 

\begin{equation}\label{vvmt_thld_critical}
	VVMT_{thld}^E = \frac{2(R_{adv}-R_{norm})}{C^j_{vvmt}-C^i_{vvmt}}-M
\end{equation}

Although a malicious vehicle can gain more rewards than a normal vehicle, the gained rewards are of the same order of magnitude for both parties. By letting $VVMT_{thld}$ $>$ $VVMT_{thld}^E$, we have a better chance of filtering out more malicious vehicles than normal vehicles, as shown in Fig~\ref{fig:exp_payoff_b}.

\textbf{$R_{norm} \ll R_{adv}$}: a malicious vehicle always have more incentives than its normal counterpart to opt-in regardless of the value of $VVMT_{thld}$. Although this exceptional situation contradicts the assumption that this paper makes on rational driver behaviors (i.e., bounded rewards and resources), it can still occur in rare cases if a malicious vehicle has unbounded resources and rewards to opt in, as shown in Fig.~\ref{fig:exp_payoff_c}. For example, an international terrorist group has identified the vehicular network in a given country as critical infrastructure and makes every effort to compromise it. We will tackle this type of malicious driver with unlimited adversarial reward in our future work. 

To summarize, as long as drivers are rational and aim to maximize their payoffs, we can filter out more malicious vehicles than normal vehicles by setting $VVMT_{thld}$ to a relatively high but reasonable value.
\end{proof}


%




\ifCLASSOPTIONcaptionsoff
  \newpage
\fi




\bibliographystyle{IEEEtran}
\bibliography{IEEEabrv.bib,root.bib}
%

%

\begin{IEEEbiography}[{\includegraphics[width=1in,height=1.25in,clip,keepaspectratio]{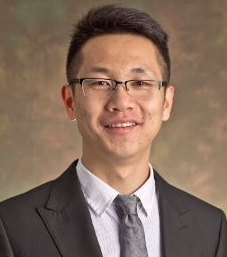}}]{Dajiang Suo} is a research scientist at MIT Auto-ID Lab. He obtained a Ph.D. in Mechanical Engineering from MIT in 2020. Suo holds a B.S. degree in mechatronics engineering, and an S.M. degree in Computer Science and Engineering Systems. His research interests include Internet of Things, connected vehicles, cybersecurity, and RFID.

Before returning to school to pursue PhD degree, Suo was with the vehicle control and autonomous driving team at Ford Motor Company (Dearborn, MI), working on the safety and cyber-security of automated vehicles. He also serves on the Standing Committee on Enterprise, Systems, and Cyber Resilience (AMR40) at the Transportation Research Board.
\end{IEEEbiography}

\vskip -2\baselineskip plus -1fil

\begin{IEEEbiography}[{\includegraphics[width=1in,height=1.25in,clip,keepaspectratio]{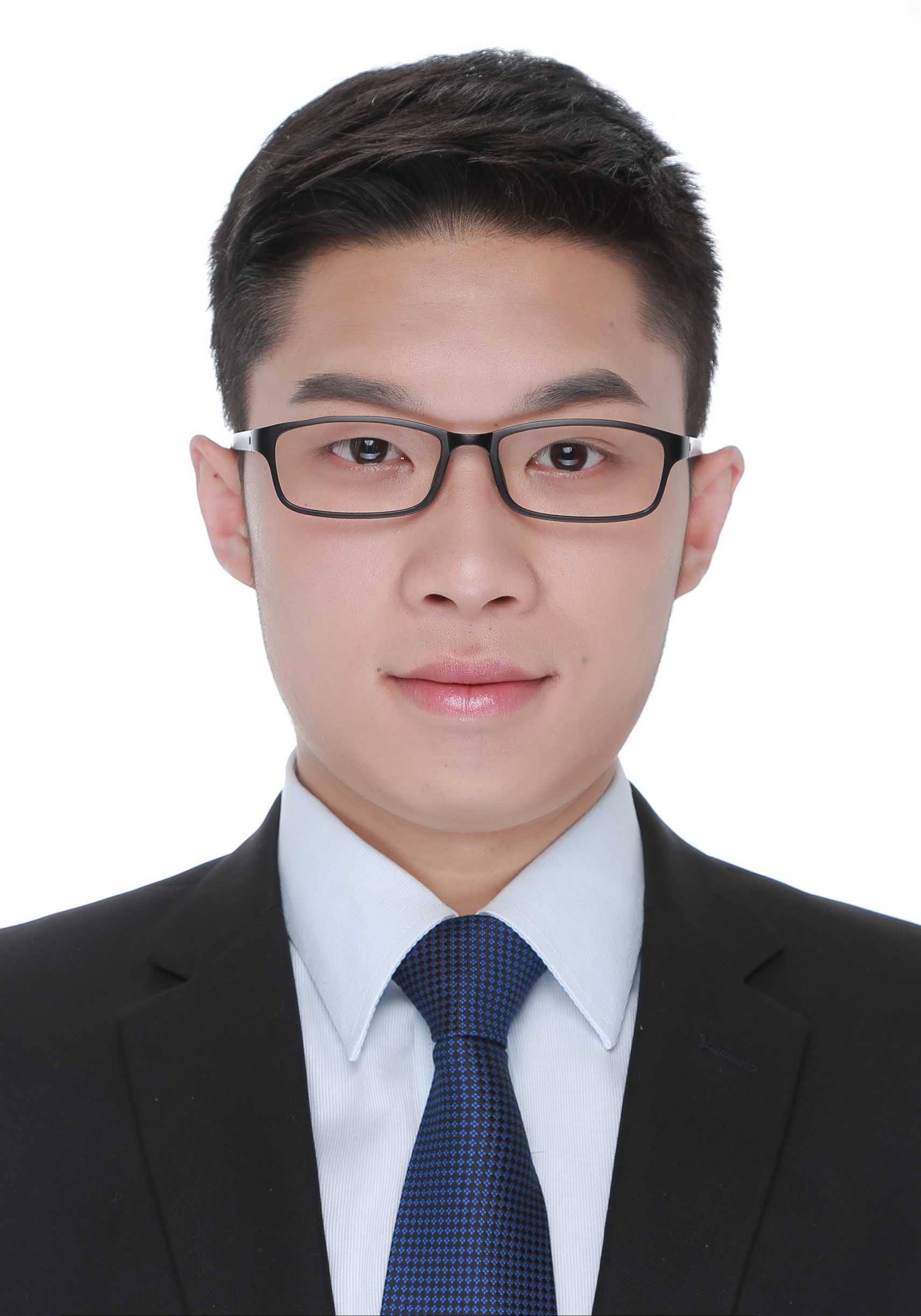}}]{Baichuan Mo} is a research scientist at Lyft. He obtained his Ph.D. degree in transportation from MIT in 2022. He holds a dual Master's degree in Transportation and Computer Science from MIT, and a Bachelor's degree in Civil Engineering from Tsinghua University. His research interests include data-driven transportation modeling, demand modeling, optimization, and applied machine learning.
\end{IEEEbiography}

\vskip -2\baselineskip plus -1fil

\begin{IEEEbiography}[{\includegraphics[width=1in,height=1.25in,clip,keepaspectratio]{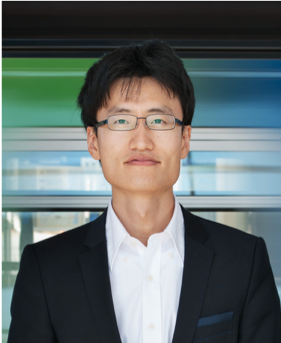}}]{Jinhua Zhao} is the Associate Professor of City and Transportation Planning at the Massachusetts Institute of Technology (MIT). Prof. Zhao brings behavioral science and transportation technology together to shape travel behavior, design mobility system, and reform urban policies. He develops methods to sense, predict, nudge, and regulate travel behavior and designs multimodal mobility systems that integrate automated and shared mobility with public transport. He sees transportation as a language to describe a person, characterize a city, and understand an institution and aims to establish the behavioral foundation for transportation systems and policies.
Prof. Zhao directs the JTL Urban Mobility Lab and Transit Lab at MIT and leads long-term research collaborations with major transportation authorities and operators worldwide, including London, Chicago, Hong Kong, and Singapore. He is the co-director of the Mobility Systems Center of the MIT Energy Initiative, and the director of the MIT Mobility Initiative.
\end{IEEEbiography}

\vskip -2\baselineskip plus -1fil

\begin{IEEEbiography}[{\includegraphics[width=1in,height=1.25in,clip,keepaspectratio]{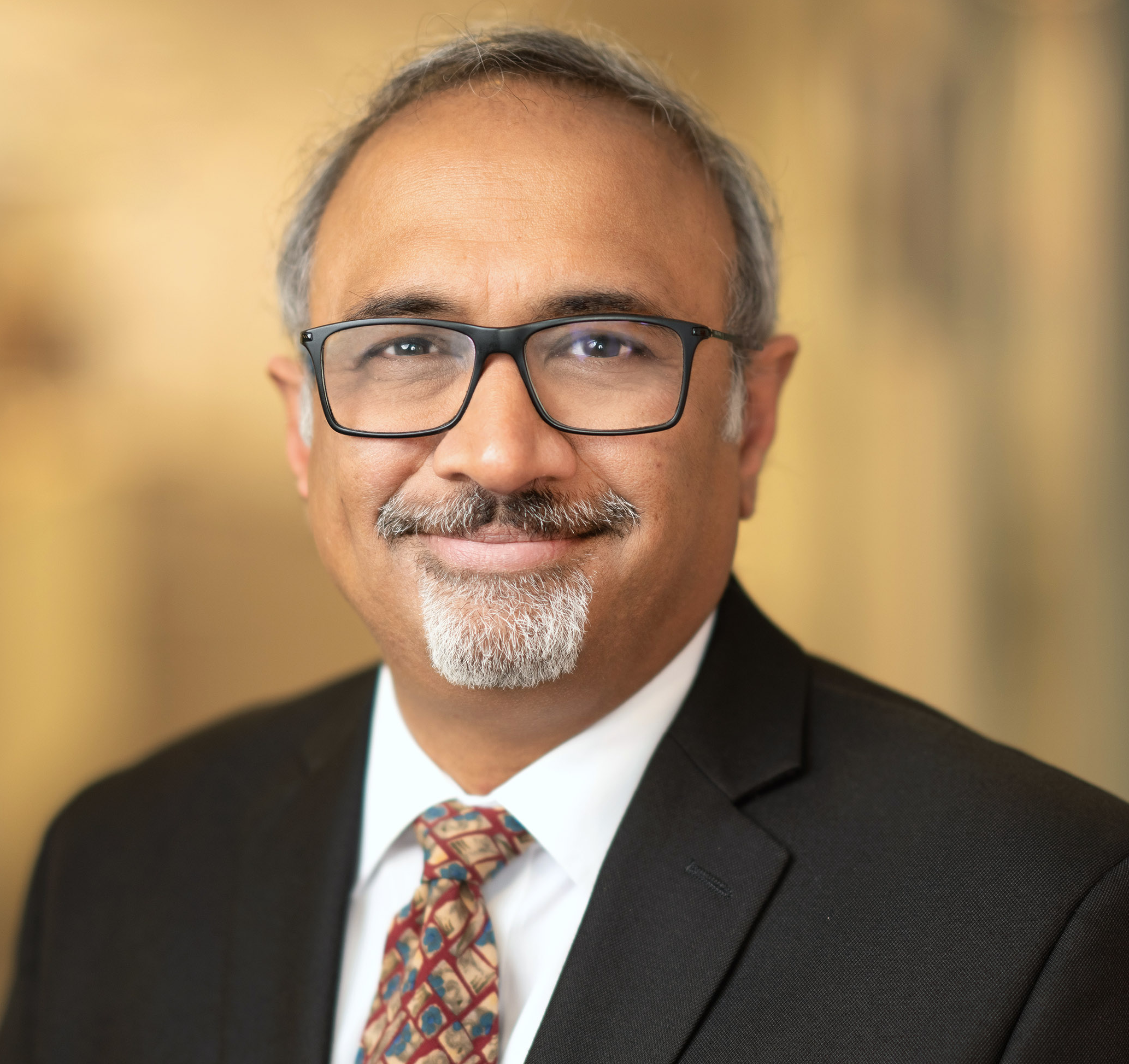}}]{Sanjay E. Sarma}
Sanjay Sarma is the Fred Fort Flowers (1941) and Daniel Fort Flowers (1941) Professor of Mechanical Engineering at MIT. He co-founded the Auto-ID Center at MIT and developed many of the key technologies behind the EPC suite of RFID standards now used worldwide. He was also the founder and CTO of OATSystems, which was acquired by Checkpoint Systems (NYSE: CKP) in 2008. He serves on the boards of GS1, EPCglobal and several companies including CleanLab and Aclara Resources (TSX:ARA). 

Dr. Sarma received his Bachelors from the Indian Institute of Technology, his Masters from Carnegie Mellon University and his PhD from the University of California at Berkeley. Sarma also worked at Schlumberger Oilfield Services in Aberdeen, UK. He has authored over 150 academic papers in computational geometry, sensing, RFID, automation and CAD, and is the recipient of numerous awards for teaching and research including the MacVicar Fellowship, the Business Week eBiz Award and Informationweek’s Innovators and Influencers Award.
\end{IEEEbiography}







\end{document}